\documentclass[12pt]{article}
\usepackage{amsmath}
\usepackage{graphicx,psfrag,epsf}
\usepackage{enumerate}
\usepackage{url} 
\usepackage{float}
\usepackage{framed}
\usepackage{algorithm}
\usepackage{algorithmic}
\usepackage{rotating}
\usepackage[table,xcdraw,dvipsnames]{xcolor}
\usepackage{subfigure}

\usepackage{algorithm,algorithmic}
\usepackage[algo2e]{algorithm2e}
\usepackage{multirow}
\usepackage{lipsum}
\usepackage{amsfonts}
\usepackage{cite}      
                  
\usepackage{graphicx}   

\usepackage{amsmath}
\usepackage{amsthm}
\usepackage{amsfonts}
\usepackage{amssymb}
\usepackage{graphicx}
\usepackage{subfigure}
\usepackage{transparent}
\usepackage{multirow}
\usepackage{color}
\usepackage{blkarray}
\usepackage[round]{natbib}
\usepackage{booktabs}
\usepackage[]{hyperref}
\hypersetup{colorlinks=true, 
	linkcolor=red,       
    citecolor=blue,        
    filecolor=magenta,      
    urlcolor=cyan           
}  

\newcommand{\bm}[1]{\boldsymbol{#1}}
\newcommand{\E}{\mathbb{E}}

\renewcommand{\P}{\mathbb{P}}

\newcommand{\Y}{\bm{Y}}
\newcommand{\bG}{\bm{\Gamma}}

\newcommand{\V}{\bm{V}}

\newcommand{\wt}[1]{\widetilde{#1}}

\def\A{\mbox{\boldmath$A$}}
\def\C{\mbox{\boldmath$C$}}

\def\V{\mbox{\boldmath$V$}}

\def\R{\mbox{\boldmath$R$}}

\def\zero{\mbox{\boldmath$0$}}

\def\bPhi{\mbox{\boldmath$\Phi$}}
\def\bSigma{\mbox{\boldmath$\Sigma$}}

\def\bUpsilon{\mbox{\boldmath$\Upsilon$}}

\def\B{\mbox{\boldmath$B$}}

\newcommand{\ind}{\stackrel{ind}{\sim}}

\def\bomega{\mbox{\boldmath$\omega$}}
\def\bepsilon{\mbox{\boldmath$\epsilon$}}

\def\bSigma{\mbox{\boldmath$\Sigma$}}

\def\e{\mbox{\boldmath$e$}}
\def\E{\mathbb{E}}

\renewcommand{\Y}{\bm{Y}}

\newcommand{\iid}{\stackrel{iid}{\sim}}

\newcommand{\wh}[1]{\smash{\widehat{#1}}}

\usepackage{array}
\newcolumntype{C}{@{\extracolsep{0.5in}}c@{\extracolsep{0pt}}}

\def\C {\,|\:}

\def\bomega{\bm{\omega}}
\def\C {\,|\:}

\def\B{\bm{B}}

\def\BGamma{\bm{\Gamma}}

\def\BD{\bm{\Delta}}
\def\BOmega{\bm{\Omega}}

\renewcommand{\e}{\mathrm{e}}

\newcommand{\N}{\mathbb{N}}
\renewcommand{\R}{\mathbb{R}}

\newtheorem{lemma}{Lemma}[section]

 \theoremstyle{assumption}

\begin{document}

\def\spacingset#1{\renewcommand{\baselinestretch}%
{#1}\small\normalsize} \spacingset{1}


  \title{{\sf 
   Dynamic Sparse Factor Analysis}}

  \author{McAlinn K.\footnote{Senior Research Professional at the University of Chicago Booth School of Business},  Rockova V.\footnote{Assistant Professor in Econometrics and Statistics and James S. Kemper Faculty Scholar at the University of Chicago Booth School of Business } and Saha E.\footnote{Ph.D. Student at the  University of Chicago Department of Statistics}}
  
   \maketitle

\bigskip
\begin{abstract}
Its conceptual appeal and effectiveness has made latent factor modeling an indispensable  tool for multivariate analysis.
Despite its popularity across many fields, there are outstanding methodological challenges that have hampered practical deployments.
One major challenge is the selection of the number of factors, which is exacerbated for  {\sl dynamic} factor models, where  factors can  disappear, emerge, and/or reoccur over time. Existing tools that assume a fixed number of factors may provide a misguided representation of the data mechanism, especially when the number of factors is crudely misspecified.
 Another challenge is the interpretability of the factor structure, which is often regarded as an unattainable objective due to the lack of identifiability.
Motivated by a topical macroeconomic application, we develop a flexible Bayesian method for dynamic factor analysis (DFA) that can simultaneously accommodate  a time-varying number of factors and enhance interpretability without strict identifiability constraints.
To this end, we turn to dynamic sparsity by employing Dynamic Spike-and-Slab (DSS) priors within DFA. 
Scalable Bayesian EM estimation is proposed for  fast posterior mode identification via rotations to sparsity, enabling Bayesian data analysis at scales that would have been previously time-consuming. 
We study a large-scale balanced panel of macroeconomic variables covering multiple facets of the US economy, with a focus on the Great Recession, to highlight the efficacy and usefulness of our proposed method.
\end{abstract}

\noindent%
\vfill

\newpage
\spacingset{1.45} 



\section{Introduction \label{sec:intro}}
The premise of dynamic factor analysis (DFA) is fairly straightforward:  there are unobservable commonalities in the variation of observable time series, which can be exploited for interpretation, forecasting, and decision making.
Dating back to, at least, \cite{burns1947measuring}, the fundamental idea that a small number of indices drive co-movements of many time series has found plentiful empirical support across a wide range of applications including economics \citep{stock2002forecasting,bai2002determining,bernanke2005measuring,baumeister2010changes,cheng2016shrinkage}, finance \citep{diebold1989dynamics,aguilar1998bayesian,pitt1999time,aguilar2000bayesian,carvalho2011dynamic}, and ecology \citep{zuur2003estimating}, to name just a few. 
More notably, in their  seminal work  on DFA, \cite{sargent1977business} showed that two dynamic factors could explain a large fraction of the variance of U.S. quarterly macroeconomic variables.
Motivated by a similar (but significantly larger) application, we develop scalable Bayesian DFA  methodology and deploy it to  glean insights into the hidden drivers of the U.S. macroeconomy before, during and after  the Great Recession. 

With  large-scale cross sectional data becoming readily available, the need for developing scalable and reliable tools adept at capturing complex latent dynamics 
have spurred in both statistics  and econometrics \citep{beyeler2016factor,kaufmann2017identifying,fruehwirth2018sparse,nakajima2017dynamics}.
While ``dynamic factor models have been the main big data tool used over the past 15 years by empirical macroeconomists" \citep{stock2016dynamic}, there are remaining methodological challenges. It is now commonly agreed that high-dimensional inference can hardly be formalized and executed without any sparsity assumptions. 
The fundamental goal of our research is to facilitate sparsity  {\sl discovery} (i.e. data-informed sparsity), when in fact present. 
In doing so, we keep in mind  three main pillars that we regard as essential for building a stable foundation  for sparse factor modeling.

Firstly, the latent factor loadings should account for time-varying patterns of sparsity.
In (macro-)economics and finance, the sequentially observed  variables may go through multiple periods of shocks, expansions, and contractions \citep{hamilton1989new}.
It is thus expected that the underlying latent structure  changes over time-- either gradually or suddenly-- 
where  some factors might be active at all times, while others only at certain times.
For example, in our empirical analysis we find  that certain factors exert influence on some series only during a crisis and later permeate through different components of the economy as the shock spreads.
Dynamic sparsity plays a very compelling role in capturing and characterizing such dynamics. 
Recent developments in sparse factor analysis reflect this direction of interest \citep{West2003,carvalho2008high,yoshida2010bayesian,lopes2010cholesky}.
More recently, \cite{nakajima2012dynamic} deployed the latent threshold approach of \cite{nakajima2013bayesian} in order to induce zero loadings dynamically over time.
Our methodological contribution builds on this development, but poses far less practical limitations on the dimensionality of the data and far less constraints on identification.

Related to the previous point is the question of selecting the number of factors.
This modeling choice is traditionally determined by a combination of {\it a priori} knowledge, a visual inspection of the scree plot \citep{onatski2009testing}, and/or information criteria \citep{bai2002determining,hallin2007determining}.
In the presence of model uncertainty, the Bayesian approach affords the opportunity to assign a probabilistic blanket over various models. Bayesian non-parametric approaches have been considered for estimating the factor dimensionality using sparsity inducing priors \citep{bhattacharya2011sparse,rovckova2016fast}.
The added difficulty stemming from time series data, however, is that the number of factors {\sl may change over time}. 
Despite plentiful empirical evidence for this behavior in macroeconomic data \citep{bai2002determining},  the majority of existing DFA tools treat the number of factors as fixed over time.
As a remedy, we turn to dynamic sparsity as a compass for determining  the number of factors without necessarily committing to one fixed number ahead of time.


The third essential requirement is accounting for structural instabilities over time with  time-varying loadings and/or factors. One seemingly simple solution has been  to deploy rolling/extending window approaches to obtain pseudo-dynamic loadings.
These estimates, however, lack any supporting probabilistic structure that would induce  smoothness and/or capture sudden dynamics.
Recent DFA developments \citep{del2008dynamic,nakajima2013bayesian}
have   treated both the factors and loadings as stochastic and  dynamic. 
Adopting this point of view,  we blend  smoothness with sparsity  via  Dynamic Spike-and-Slab  (DSS) priors on factor loadings \citep{rockova2018dynamic}.
This prior regards factor loadings  as arising from a mixture of two states: an inactive state represented by very small loadings and an active state  represented by smoothly evolving large loadings.
The mixing weights between these two states themselves are time-varying, reflecting past information  to prevent from erratic regime switching.
The DSS priors allow latent factors to effectively, and smoothly, appear or disappear from each series, tracking the evolution of sparsity over time.

In this work, we develop methodology for sparse dynamic factor analysis that is built on the three foundational principles mentioned above.
Using this methodology, we examine a large-scale balanced panel of macroeconomic indices that span multiple corners of the  U.S. economy from 2001 to 2015.
Our method  helps understand how the economy evolves over time and how shocks affect its individual components.
In particular, examining the  latent factor structure before, during, and after the Great Recession, we obtain insights into the channels of dependencies  and we assess permanence of structural changes.

To ensure that our implementation scales with large datasets, we propose an EM algorithm for MAP estimation that recovers evolving sparse latent structures in a fast and potent manner.
As the EM algorithm finds a likely sparse structure, it does not require strong identification constraints that would be needed for MCMC simulation. 
While interpretation can be achieved with ex-post rotations \citep{bai2013principal,kaufmann2017identifying}, here we deploy rotations to sparsity {\sl inside} the EM algorithm along the lines of \cite{rovckova2016fast} to (a) accelerate convergence and (b) obtain better oriented sparse solutions.

The paper is structured as follows.
Section~\ref{sec:method} outlines the dynamic sparse factor model.
Section~\ref{sec:comp} summarizes our EM estimation strategy.
A detailed simulation study that highlights our strategy relative to other methods is  in Section~\ref{sec:simstudy}.
An empirical study on a large-scale macroeconomic dataset is in Section~\ref{sec:study}.
We conclude the paper with additional comments in Section~\ref{sec:conclusion}. 
Details of the implementation are in the Supplementary Materials.

\section{Dynamic Sparse Factor Models \label{sec:method} }
The data setup under consideration consists of a matrix of high-dimensional multivariate time series $\Y=[\Y_1,\dots, \Y_T]\in \mathbb{R}^{P\times T}$, where each  vector $\Y_t\in\R^P$ contains a snapshot of continuous measurements at time $t$.
Dynamic factor models are built on the premise that  there are only a few latent factors that drive the co-movements of $\Y_t$.
Evolving covariance patterns of time series  can be captured with the following state space model:
\begin{align}
	\Y_t &= \B_t\bomega_t+\bepsilon_t, \quad \bepsilon_t\ind \mathcal{N}_P(\zero,\bSigma_t),\label{eq:lf1}\\
	\bomega_t &=\bPhi\bomega_{t-1} +\bm\e_t, \quad \bm\e_t\ind \mathcal{N}_K(\zero,\,\sigma^2_\omega\mathbb{I}_K)\label{eq:lf2},
\end{align}
which extends the more standard dynamic factor models \citep{sargent1977business,geweke1977} in at least two ways.
First, the  observation equation \eqref{eq:lf1} links  $\Y_t$ to a vector of factors $\bomega_t$ through multivariate regression  with  loadings $\B_t\in \mathbb{R}^{P\times K}$ and with residual variances $\bSigma_t=\mathrm{diag}\{\sigma_{1t}^2,\dots,\sigma_{Pt}^2\}$, where {\sl both} $\B_t$ and $\bSigma_t$ are {\sl dynamic}, i.e. are allowed to  evolve over time. In this section, we  tacitly assume that any location shifts in  $\Y$ have been standardized away and thereby we omit an intercept in \eqref{eq:lf1}. The (dynamic) intercept can be however included, as we demonstrate in Section \ref{sec:study}.
Second, the transition equation \eqref{eq:lf2}  describes the unobserved regressors  $\bomega_t$ as following a stationary autoregressive process with a transition matrix $\bPhi=\wt\phi\,\mathbb{I}_K$ for some $0<\wt\phi<1$ and with Gaussian disturbances $\bm\e_t$ with a known variance $\sigma^2_\omega>0$. 
As is customary with state-space models of this type,  we assume that $\bomega_t,$ $\bm\e_t$ and $\bepsilon_t$ are cross-sectionally independent.

A related approach was proposed in \cite{aguilar2000bayesian} and \cite{lopes2007factor}, who also permit time-varying loadings, but do not impose the AR(1) process on the factors. Instead, their factors are cross-sectionally independent and linked over time through a stochastic volatility evolution of their idiosyncratic variances.
\cite{bai2002determining}  and \cite{stock2010modeling}, on the other hand, assume that factors follow vector autoregression, but the loadings are constant over time. As in \cite{nakajima2012dynamic}, our model \eqref{eq:lf1} and \eqref{eq:lf2}   differs from these more standard dynamic factor model formulations because it combines the AR(1) factor aspect together with dynamic loadings.


The equations \eqref{eq:lf1} and \eqref{eq:lf2} imply that, marginally, $\Y_t\sim\mathcal{N}_P(0,\wt\bSigma_t)$, where $\wt\bSigma_t=\sigma^2_{\omega}/(1-\wt\phi^2)\B_t\B_t'+\bSigma_t$. This decomposition provides a fundamental justification for factor-based dynamic covariance modeling. The information in  high-dimensional vectors $\Y_t$ is distilled through latent factors into lower-dimensional factor loadings matrices $\B_t$, which {\sl completely} characterize the movements of covariances  over time.  Other authors \citep{del2008dynamic,lopes2007factor} consider a stochastic volatility (SV) evolution (either log-AR(1) or Bayesian discounting) on the variance of the latent factors and/or the innovations $\bepsilon_t$ in  \eqref{eq:lf1}. While both are feasible within our framework, here we impose Bayesian discounting SV formulation on the innovation variances: $\sigma_{jt}=\sigma_{jt-1}\delta/\upsilon_{jt},$ where $\delta\in(0,1]$ is a discount parameter and where $\upsilon_{jt}\sim\mathcal{B}(\delta\eta_{t-1}/2,(1-\delta)\eta_{t-1}/2)$ with $\eta_t=\delta\eta_{t-1}+1$ \citep[Ch. 4.3.7][]{Prado2010}.

Parsimonious covariance estimation is only one of the objectives of dynamic factor modeling. The more traditional objective is disentangling the covariance structure and understanding its driving forces  and how they change over time. {\sl Sparse} modeling has been indispensable for both of these objectives, where fewer estimable coefficients  yield far more stable covariance estimates and where nonzero patterns in $\B_t$ yield superior interpretable characterizations \citep{carvalho2008high,yoshida2010bayesian}. Next, we explore the role of dynamic sparsity in DFA.

\subsection{Dynamic Sparsity with  Shrinkage Process Priors}\label{sec:priors}
No  assumption has been as pervasive in the analysis of high-dimensional data as the one of sparsity.
Sparsity is a practical modeling choice that facilitates high-dimensional inference and/or computation. In factor model contexts, it can also be used to anchor on identifiable parametrizations \citep{fruhwirth2010parsimonious} and/or for estimating factor dimensionality \citep{rovckova2016fast,bhattacharya2011sparse}. The potential  of sparsity  in dynamic factor models has begun to be recognized \citep{nakajima2012dynamic,beyeler2016factor,kaufmann2017identifying}.




In this work, we complement the factor model formulation \eqref{eq:lf1} with dynamic sparsity priors on the factor loadings $\B_t$ for $1\leq t\leq T$. In other words, rather than imposing a dense model by assigning  a random walk (or a stationary autoregressive) prior on the loadings \citep[such as][]{stock2002forecasting,del2008dynamic}, we allow for the possibility that the loadings are zero at certain times.

We will write $\B_t=(\beta_{jk}^t)_{j,k=1}^{P,K}$ and impose a shrinkage process prior on individual time series $\{\beta_{jk}^t\}_{t=1}^T$ for each $(j,k)$.
A few authors have reported on the benefits of  dynamic variable selection in the analysis of macroeconomic data \citep{fruhwirth2010stochastic,bitto2016achieving,lopes2010cholesky,nakajima2012dynamic,koop2010bayesian}. We build on one of the more recent developments, the Dynamic Spike-and-Slab (DSS) priors proposed by \cite{rockova2018dynamic}.

DSS priors are dynamic extensions of spike-and-slab priors for variable selection \citep{george1993variable,rovckova2016spike}.
Each coefficient in DSS is thought of as arising from  two latent states: (1) an {\sl inactive} state, where the coefficient meanders randomly around zero,  and  (2) an {\sl active} state, where the coefficient walks on an autoregressive path. The switching between these two states is driven by a {\sl dynamic} mixing weight which depends on past values of the series,  making the states less erratic over time.


We begin by reviewing the conditional specification of the DSS prior.  For each coefficient $\beta_{jk}^t$, we have a binary indicator $\gamma_{jk}^t\in\{0,1\}$, which encodes the state of $\beta_{jk}^t$ (the ``spike" inactive state for $\gamma_{jk}^t=0$ and the ``slab" active state for $\gamma_{jk}^t=1$). Given  $\gamma_{jk}^t$ and a lagged value $\beta_{jk}^{t-1}$, we assume a conditional mixture prior (independently for each $(j,k)$):
\begin{equation}\label{betas}
\pi(\beta_{jk}^t|\gamma_{jk}^t,\beta_{jk}^{t-1})=(1-\gamma_{jk}^t)\psi_0(\beta_{jk}^t|\lambda_0)+\gamma_{jk}^t\psi_1\left(\beta_{jk}^t\,|\,\mu(\beta_{jk}^{t-1}),\lambda_1\right),
\end{equation}
where
\begin{equation}\label{mu}
\mu(\beta_{jk}^{t-1})=\phi_{0}+\phi_{1}(\beta_{jk}^{t-1}-\phi_0)\quad\text{with}\quad |\phi_1|<1
\end{equation}
and
\begin{equation}\label{gammas}
\P(\gamma_{jk}^t=1|\beta_{jk}^{t-1})=\theta_{jk}^t.
\end{equation}
The conditional prior \eqref{betas} is a mixture of two components:  (i) a spike Laplace density $\psi_0(\beta|\lambda_0)$ that is concentrated around zero and (ii)
 a Gaussian slab density $\psi_1(\beta_{t}|\mu(\beta_{jk}^{t-1}),\lambda_1)$, which is moderately peaked around its mean $\mu(\beta_{jk}^{t-1})$ with variance $\lambda_1$. 
This mixture formulation is an extension of existing continuous spike-and-slab priors \citep{george1993variable,ishwaran2005spike,rovckova2018bayesian}, allowing the mean $\mu(\beta_{jk}^{t-1})$ of the non-negligible coefficients  to evolve smoothly over time (through a stationary autoregressive process).
The spike distribution $\psi_0(\beta_{t}|\lambda_0)$, on the other hand, {\sl does not} depend on $\beta_{jk}^{t-1}$, effectively shrinking the negligible coefficients towards zero.
In this regard, the conditional prior in \eqref{betas}  can be seen as a ``multiple shrinkage" prior \citep{george1986minimax,george1986combining} with two centers of gravity.

In time series data (as will be seen from our empirical study), it reasonable to expect that some factors are active only for some periods of time.  Such ``pockets of predictability" \citep{farmer2018pockets} can be captured with  spike/slab memberships $\gamma_{jk}^t$ that evolve somewhat smoothly.
This behavior can be encouraged with dynamic mixing weights $\theta_{jk}^t$  (defined in \eqref{gammas}) that  reflect past information.
To this end, we deploy the deterministic construction of \cite{rockova2018dynamic} defined, for some global balancing parameter $0<\Theta<1$,   as follows
\begin{equation}\label{weights}
\theta_{jk}^t\equiv\theta(\beta_{jk}^t)=\frac{\Theta\psi_1^{ST}\left(\beta_{jk}^t|\lambda_1,\phi_0,\phi_1\right)}{\Theta\psi_1^{ST}\left(\beta_{jk}^t|\lambda_1,\phi_0,\phi_1\right)+(1-\Theta)\psi_0\left(\beta_{jk}^t|\lambda_0\right)},
\end{equation}
given $(\Theta,\lambda_0,\lambda_1,\phi_0,\phi_1)$. This mixing weight has an interesting interpretation. It is defined as the {\sl marginal} inclusion probability  $\P(\gamma_{jk}^{t-1}=1\C\beta_{jk}^{t-1})$ for classifying $\beta_{jk}^{t-1}$ as arising from the {\sl stationary} slab distribution $\psi_1^{ST}\left(\beta_{jk}^t|\lambda_1,\phi_0,\phi_1\right)$, as opposed to the stationary spike distribution
$\psi_0\left(\beta_{jk}^t|\lambda_0\right)$, under the prior $\P(\gamma_{jk}^{t-1}=1)=\Theta$. 
As $\theta_{jk}^t$'s evolve over time,  they project the latent state (active/inactive) of the past value onto the next values.
These weights induce marginal stability in the sense that each coefficient $\beta_{jk}$ has a {\sl marginal spike-and-slab distribution}, i.e.
$\pi(\beta_{jk})=\Theta \psi_1^{ST}\left(\beta_{jk}^t|\lambda_1,\phi_0,\phi_1\right)+(1-\Theta)\psi_0\left(\beta_{jk}^t|\lambda_0\right)$ \citep[see Theorem 1 of][]{rockova2018dynamic}.

Having introduced the DSS priors, we can now fully specify our dynamic latent factor model with \eqref{eq:lf1}, \eqref{eq:lf2}, \eqref{betas}, \eqref{mu} and \eqref{gammas}. The autoregressive parameters $\phi$ and $\wt\phi$ are set fixed to values close to $1$.  
Our sparse dynamic factor model is related to the  approach of  \cite{nakajima2012dynamic}, who zero out loadings whenever their autoregressive path drops bellow a certain threshold \citep[see][for comparisons]{rockova2018dynamic}. Another related approach is by  \cite{beyeler2016factor}, who induce a point-mass spike and slab prior on the loadings. However, their approach (a) does not link the inclusion indicators and loadings over time, and (b)  MCMC is deployed for calculations. Here, we develop an EM estimation procedure which does not require strong identifiability constraints.

\subsection{Identifiability Considerations}
Factor models are not free from identifiability problems, owing to the fact that the model \eqref{eq:lf1} and \eqref{eq:lf2} is observationally equivalent to $\Y_t=\B^\star_t\bomega^\star_t+\bepsilon_t$ and $\bomega^\star_t=\bm{\Phi}\bomega^\star_{t-1}+\bm e_t$, where $\bomega^\star_t=\A_t\bomega_t$ and $\B^\star_t=\B_t\A_t'$ for any orthonormal matrix $\A_t$. To ensure identifiability, it is customary to restrict $\B_t$ to be lower-triangular, with ones on the diagonal \citep{nakajima2012dynamic,aguilar2000bayesian,lopes2004bayesian,lopes2007factor} or some variant of this form \citep{fruhwirth2010parsimonious}. Nevertheless, these constraints render the analysis ultimately dependent on the ordering of the responses. Identification restrictions are particularly important for Bayesian analysis with MCMC, where meaningful interpretation of $\B_t$ is hampered by averaging over various model orientations in the Markov Chain. Our approach, although conceptually Bayesian, {\sl does not} rely on MCMC, but instead deploys optimization for posterior mode finding. In this vein, identifiability is less of a concern and can be even taken advantage of for 
mode jumping \citep{rovckova2016fast}. We thus do not induce any strict identifiability constraints besides the requirement that each nonzero column $\B_t$ has to contain at least two nonzero entries.

\subsection{Estimating Factor Dimensionality}
The factor model \eqref{eq:lf1} and \eqref{eq:lf2} is formulated conditionally on the number of factors $K\in\N$. As noted by \cite{bai2002determining}, ``the correct specification of the number of factors is central to both the theoretical and empirical validity of factor models." The authors propose a criterion and show that it is consistent for estimating $K$ in high-dimensional setups. In another strand of research, sparsity has  been exploited for determining the effective factor dimensionality \citep{fruhwirth2010parsimonious}. In particular, Bayesian non-parametric formulations have been proposed \citep{bhattacharya2011sparse,rovckova2016fast}, where $K$ is extended to infinity, while making sure that the number of {\sl nonzero} columns in $\B_t$ is finite with probability one. Treating $K$ as random in this way under sparsity priors (such as those discussed in Section \ref{sec:priors}), the posterior output
can be used to determine $K$. We adopt a similar approach  to \cite{rovckova2016fast}, where $K$ in \eqref{eq:lf1} is purposefully over-estimated and the number of {\sl nonzero} columns obtained under strict sparsity priors will indicate how many {\sl effective} factors there are.

\section{Estimation Strategy \label{sec:comp} }
To estimate the proposed dynamic latent factor model with DSS priors, we use the EM algorithm \citep{dempster1977maximum}, which  allows for fast identification of posterior modes by iteratively maximizing the conditional expectation of the log posterior. The EM algorithm is well-suited for latent variable models, such as  factor analysis, where it has been deployed by multiple authors including \cite{rubin1982algorithms,watson1983alternative,zuur2003estimating} and, more recently, \cite{rovckova2016fast}. EM can be motivated by two simple facts. First, if we knew the missing data, standard estimation techniques
can be deployed to estimate model parameters. Second, once we update our beliefs about model parameters we can make a much better educated guess about the missing data. Iterating between these two steps provides a fast way of obtaining maximum likelihood estimates and posterior modes.

Our EM algorithm has a few extra features that make it particularly attractive for dynamic factor analysis. First, the DSS priors (with a Laplace spike at zero) create spiky posteriors with sparse modes at coordinate axes. These modes yield interpretable   latent factor structures that are anchored on sparse representations without arbitrary identifiability constraints.
Second, the number of {\sl active} factors does not have to be pre-specified and can be inferred from the dynamically evolving sparse structure.


{
\spacingset{0.8}
\begin{table}[!t]
\small
\begin{center}
\scalebox{0.8}{
\begin{tabular}{|lll|}
\hline
\multicolumn{3}{|c|}{\cellcolor[HTML]{C0C0C0} \textbf{Algorithm:} \textit{EM algorithm for Automatic Rotations to Sparsity}} \\ \hline\hline
\multicolumn{1}{|c}{}                 & \multicolumn{2}{l|}{Initialize $\BD=(\bm{B}_{0:T},\bSigma_{1:T})$}                                \\
                                      & \multicolumn{2}{l|}{Repeat the following E-Step, M-Step and Rotation step until convergence}       \\
\multicolumn{3}{|c|}{\cellcolor[HTML]{C0C0C0}The E-Step}                                                  \\
 &     & {For $t=1,\dots,T$}\\
E1:             & Latent Features:                   & Get $\bomega_{t\C T}, \V_{t\C T}$ and $\V_{t,t-1 \C T}$ from the Kalman filter and smoother                                 \\
E2:             & Latent Indicators                  & Compute $\langle \gamma_{jk}^t \rangle$ for $j=1, \dots, P$, $k=1,\dots,K $,                                             \\ 
& & $\langle \gamma_{jk}^0 \rangle  = \frac{\Theta \psi_1(\beta_{jk}^{0}|0,\frac{\lambda_1}{1-\phi^2})}{\Theta \psi_1(\beta_{jk}^{0}|0,\frac{\lambda_1}{1-\phi^2})+ (1-\Theta) \psi_0(\beta_{jk}^{0}|0,\lambda_0)}$ \\
& & $\langle \gamma_{jk}^t \rangle  = \frac{\theta_{jk}^t \psi_1(\beta_{jk}^{t}|\phi \beta_{jk}^{t-1},\lambda_1)}{\theta_{jk}^t \psi_1(\beta_{jk}^{t}|\phi \beta_{jk}^{t-1},\lambda_1)+ (1-\theta_{jk}^t) \psi_0(\beta_{jk}^{t}|0,\lambda_0)}$\\
\multicolumn{3}{|c|}{\cellcolor[HTML]{C0C0C0}The M-Step}                                                  \\
M1: &  Loadings   & {For $t=0,\dots,T$} \\
       &       & Update $\beta_{jk}^{t*}$, for $j=1, \dots, P$, $k=1,\dots,K$ following \eqref{update_betas}.                                                                \\
M2:  & Rotation Matrix    & Set $\bm A_0=\bm I_K$\\
& &  {For $t=1,\dots,T$} \\
            &                     & Update                                        $\A_t  = \bm{M}_{1t}-\bm{M}_{12t}-\bm{M}_{12t}'+\bm{M}_{2t}$, where       \\
& & $\bm{M}_{1t}=\bomega_{t-1\C T}\bomega_{t-1\C T}'+\V_{t-1 \C T}$ \\
& & $\bm{M}_{12t}=\bomega_{t-1\C T}\bomega_{t\C T}'+\V_{t,t-1 \C T}$ \\
& & $\bm{M}_{2t}=\bomega_{t \C T}\bomega_{t \C T}'+\V_{t \C T}$ \\
{M3:}             & Idiosyncratic Variance             & {Compute $\bSigma_{1:T}$ using Forward Filtering Backward Smoothing }                                              \\
\multicolumn{3}{|c|}{\cellcolor[HTML]{C0C0C0}The Rotation Step}                                           \\
R: & Rotation & {For $t=0, \dots, T$} \\
              &                            & Get Cholesky decomposition $\A_t=\A_{tL} \A_{tL}'$                                          \\
& & {Rotate $\B_t=\B_t^{*} \A_{tL}$} \\ \hline
\end{tabular}}
\end{center}
\caption{\small Parameter Expanded EM algorithm for sparse Bayesian dynamic factor analysis 
}
\label{EM}
\end{table}
}

As we discussed in Section 2.2, the model is invariant under rotation of  factor loading matrices. While this lack of identifiability has  been regarded as a setback,
it can be regarded as an opportunity. Rotational invariance creates ridge-lines in the posterior that connect posterior modes and that can guide optimization trajectories \citep{rovckova2016fast}. We follow the parameter expansion approach 
\citep[see also][]{meng1992performing,liu1999parameter} that intentionally
over-parametrizes the model and takes advantage of the lack of identification to speed up convergence.
{Similarly as \cite{rovckova2016fast}, we  work with the expanded model
\begin{align}
	\Y_t &= \B_t\A_{tL}^{-1}\bomega_t+\bepsilon_t, \quad \bepsilon_t\ind \mathcal{N}_P(\zero,\bSigma_t),\label{eq:lf1_ex}\\
	\bomega_t &=\bPhi\bomega_{t-1} +\bm\e_t, \quad \bm\e_t\ind \mathcal{N}_K(\zero,\A_t)\label{eq:lf2_ex},
\end{align}
where $\A_{tL}$ is the lower Cholesky factor of a positive semi-definite matrix $\A_t$ and $\A_t \stackrel{i.i.d}{\sim} \pi(\A)\propto 1$. {We assume the initial condition $\bm\omega_0 \sim \mathcal{N}_K(\bm 0, \sigma^2_\omega/(1-\wt{\phi}^2)\bm I_K)$}
and impose the DSS prior on the individual entries of the {\sl rotated} matrix $\B_t^\star=\B_t\A_{tL}^{-1}$. The idea is to rotate towards sparse orientations throughout the iterations of the EM algorithm. The key
observation is as follows: while matrices $\A_t$ for $1\leq t\leq T$  {\sl cannot} be identified from the observed data $\Y$, they can be identified from the complete data. {Both   $\BOmega=[\bomega_0,\dots,\bomega_T]$ and $\bG=[\bG_0,\dots,\bG_T]$ are treated as the missing data. The reduced model is obtained by setting $\A_t=\sigma^2_\omega\bm I_K$ for all $1\leq t\leq T$}.} 


{{Let us denote $\BD= (\B_0,\B_{1:T},\bSigma_{1:T})$ the model parameters.} The matrix $\B_0$ contains the initial conditions that are assumed to arise from the stationary spike-and-slab prior distribution (similarly as in \cite{rockova2018dynamic}) and  $\B_{1:T}$ denotes all matrices $\B_t$ for $1\leq t\leq T$.
The goal of the EM algorithm is to find parameter values $\wh{\BD}$, which are most likely ({\it a posteriori}) to have generated the data, i.e. $\wh{\BD}= \arg \max_{\BD} \log \pi(\BD \mid \Y)$. 
This is achieved indirectly by iteratively maximizing the expectation of   the augmented log-posterior, treating the hidden factors $\BOmega$ and $\BGamma$ as missing data. 
Starting with an initialization $\BD^{(0)}$, the $(m+1)^{st}$ step of the  EM algorithm outputs $\BD^{(m+1)}= \arg \max_{\BD} Q(\BD \C\BD^{(m)})$, where $Q(\BD \C\BD^{(m)})= \E_{\BGamma, \BOmega \mid \Y, \BD^{(m)}} [\log \pi (\BD, \BGamma, \BOmega \mid \Y)]$ with 
$\E_{\BGamma, \BOmega \mid \Y, \BD^{(m)}} (.)$ denoting the conditional expectation given the observed data and current parameter estimates at the $m^{th}$ iteration. The EM algorithm iterates between the E-step (obtaining the conditional expectation of the log-posterior) and the M-step (obtaining $\BD^{(m+1)}$). The {\sl parameter-expanded EM} works in a slightly different manner.
}

{The E-step of the {\sl parameter-expanded version} operates in the reduced space (keeping $\A_t=\sigma^2_\omega\bm{I}_K$), while the M-step operates in the expanded space (allowing for general $\A_t$).
 Namely, the E-step computes the  expectation $Q(\BD \C \BD^{(m)})$  with respect to the conditional distribution of $\BOmega$ and $\BGamma$ under the {\sl original model} anchoring on  $\B_t$ and $\A_t=\sigma^2_\omega\bm{I}_K$, rather than on  $\B_t^\star$ and unrestricted $\A_t$. }
{The M-step, on the other hand, is performed in the {\sl expanded} parameter space, {where optimization takes place over $\B^\star_{0:T}$, $\bSigma_{1:T}$, and $\A_{1:T}$. 
Updating $\B^{\star(m+1)}_{0:T}$ boils down to solving a series of independent  penalized dynamic regressions \citep[as in][]{rockova2018dynamic}. {The idiosyncratic variances $\bSigma_t=\mathrm{diag}\{\sigma_{1t}^2,\dots,\sigma_{Pt}^2\}$ for $t=1,\dots,T$} are estimated in the M-step using  Forward Filtering Backward Smoothing \ref{ffbs} \citep[Ch.~4.3.7][]{Prado2010} using  the discount SV specification (as discussed in the Supplemental Material). Since $\A_{1:T}$ can be inferred from the complete data, one can estimate these matrices in the M-step to leverage the information in the missing data.
Nevertheless, the updated  matrices $\A_{1:T}$ are  {\sl not} carried forward towards the next E-step (which uses $\A_t=\sigma^2_\omega\bm{I}_K$), but are  used to {\sl rotate} the solution $\B^{\star(m+1)}_{0:T}$ back towards the reduced space via $\B_t^{(m+1)}=\B_t^{\star(m+1)}\A_{tL}$.} See \cite{rovckova2016fast} for more explanations of parameter expansion for factor rotations.
The steps of the algorithm are carefully explained in Section \ref{sec:mstep}. The computations are summarized in Table \ref{EM}}. The convergence of the EM algorithm with parameter expansion is provably faster \citep{meng1992performing,rovckova2016fast}.


\begin{figure}[t!]
\centering
\includegraphics[width=0.8\textwidth]{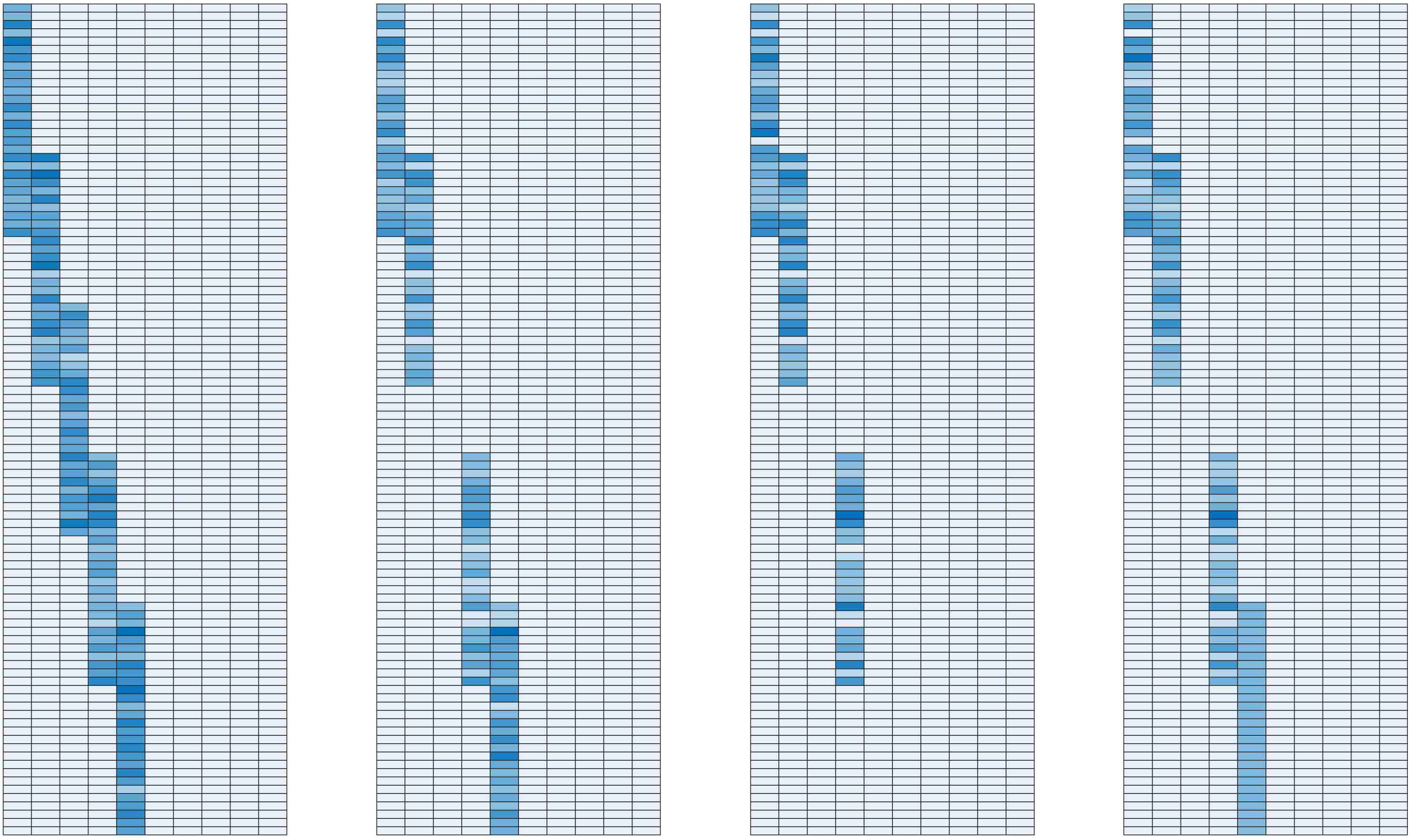}
\caption{Simulation Study:  The true latent factor loadings $\B_{t}^0$ at $t=1,101,201,301$.}
\label{fig:true}
\end{figure}

\section{Simulation Study \label{sec:simstudy}}
We illustrate the usefulness of our proposed approach, relative to multiple existing methods, on  synthetic data, reflecting the following characteristics that can occur in real  applications:  dynamic patterns of sparsity, smoothness, and  a time-varying   factor dimension. 

First, we generate a single dataset with $P=100$ responses, $K=10$ candidate latent factors, and $T=400$ time series observations (extra $100$  data points are 
generated  as training data for the rolling window analysis,  as will be described below).
The dimensionality of this example is already beyond practical limits of many Bayesian procedures.
The elements of latent factors $\BOmega_t$ and idiosyncratic errors $\bm{\epsilon}_t$ are generated from a standard Gaussian distribution.
Only the first five factors are potentially active over time, with the latter five being always inactive. We now describe the true loading matrices $\B^0=[\B^0_1,\dots,\B^0_T]$,  which were used to generate the data, where  $\B^0_t=\{\beta_{jk}^{0t}\}\in\R^{P\times K}$. At time $t=1$,  the active latent factor loadings form a block diagonal structure with $28$ active loadings per factor, of which $10$ overlap with another factor.
In other words, we have $60$ series with only one active factor, and $40$ with two active factors (see the leftmost image in Figure~\ref{fig:true}).
The sparsity pattern changes structurally over time where (a) at time $t=101$ the  loadings of the third factor become inactive, (b) at $t=201$ the loadings of the fifth factor  become inactive, and (c) at $t=301$ the loadings of the fifth factor are  re-introduced and active until $T=400$ (Figure~\ref{fig:true}).
The true nonzero loadings are smooth and arrive from an autoregressive process,  i.e.
$\beta_{jk}^{0t}=\phi\beta_{jk}^{0t-1}+v_{jk}^t$ with  $v_{jk}^t\iid \mathcal{N}(0,0.0025)$ for $\phi=0.99$, initiated at $\beta_{jk}^{01}=2$ for all $1\leq j\leq P$ and $1\leq k\leq 5$.
When loadings $\beta_{jk}^{0t}$ become inactive, they are thresholded to zero.
The true factor loadings are thereby smooth until they suddenly  drop  out and can emerge. 

\begin{figure}[htbp!]
\centering
\centering
\subfigure[$t=100$]{\includegraphics[width=12cm, height=3.5cm]{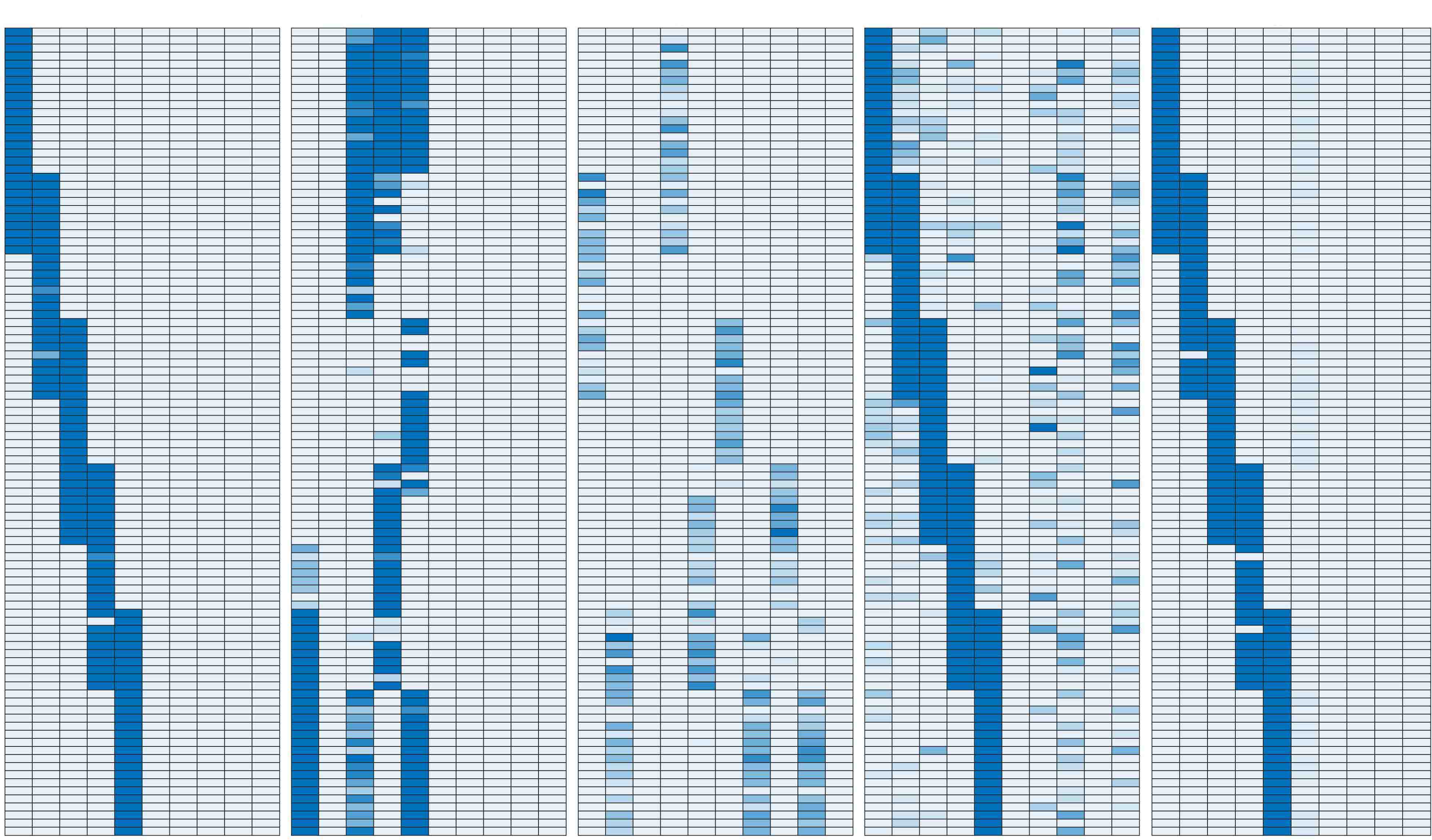}}\label{fig:t100}

\subfigure[$t=200$]{\includegraphics[width=12cm, height=3.5cm]{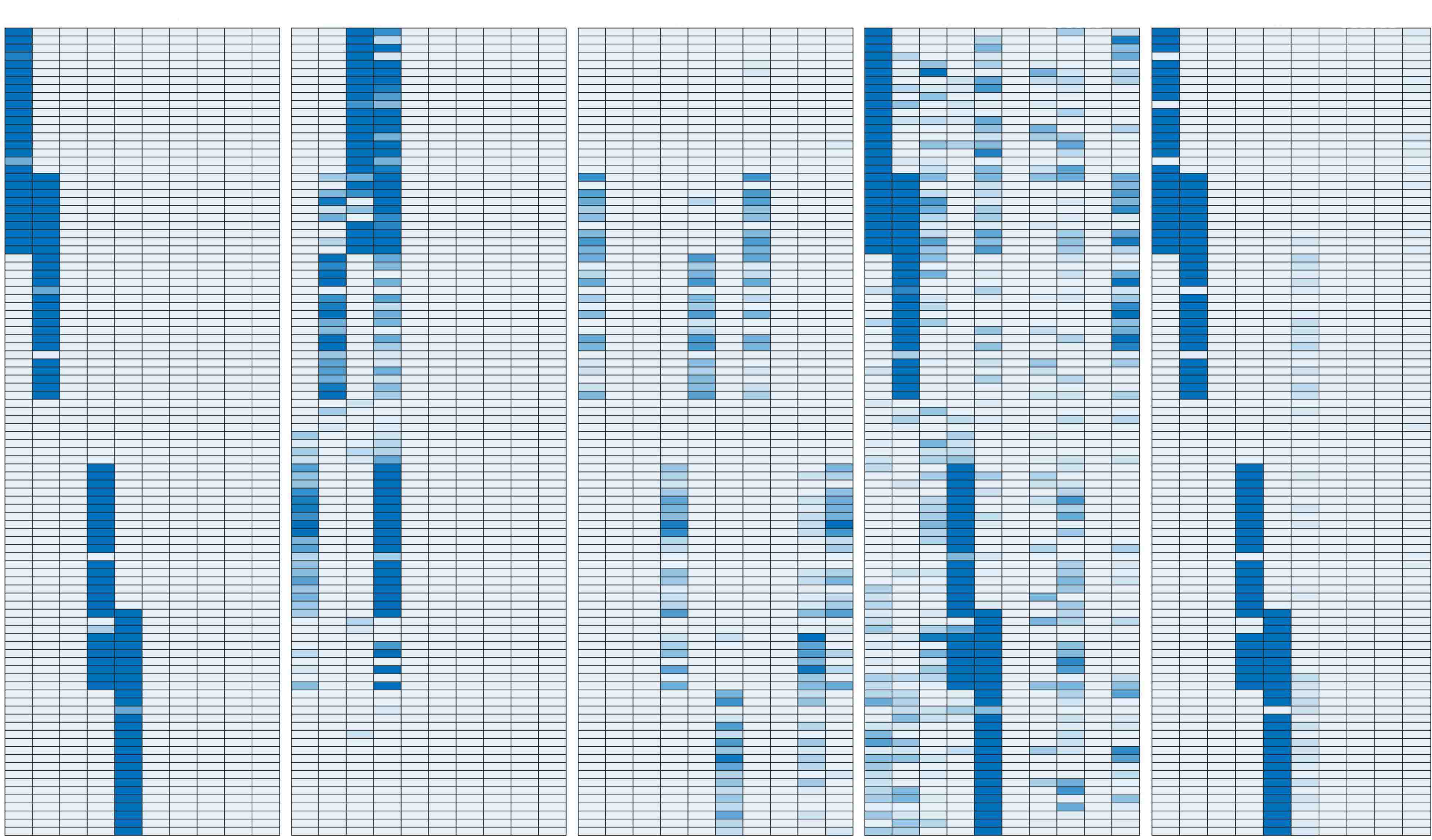}}\label{fig:t200}

\subfigure[$t=300$]{\includegraphics[width=12cm, height=3.5cm]{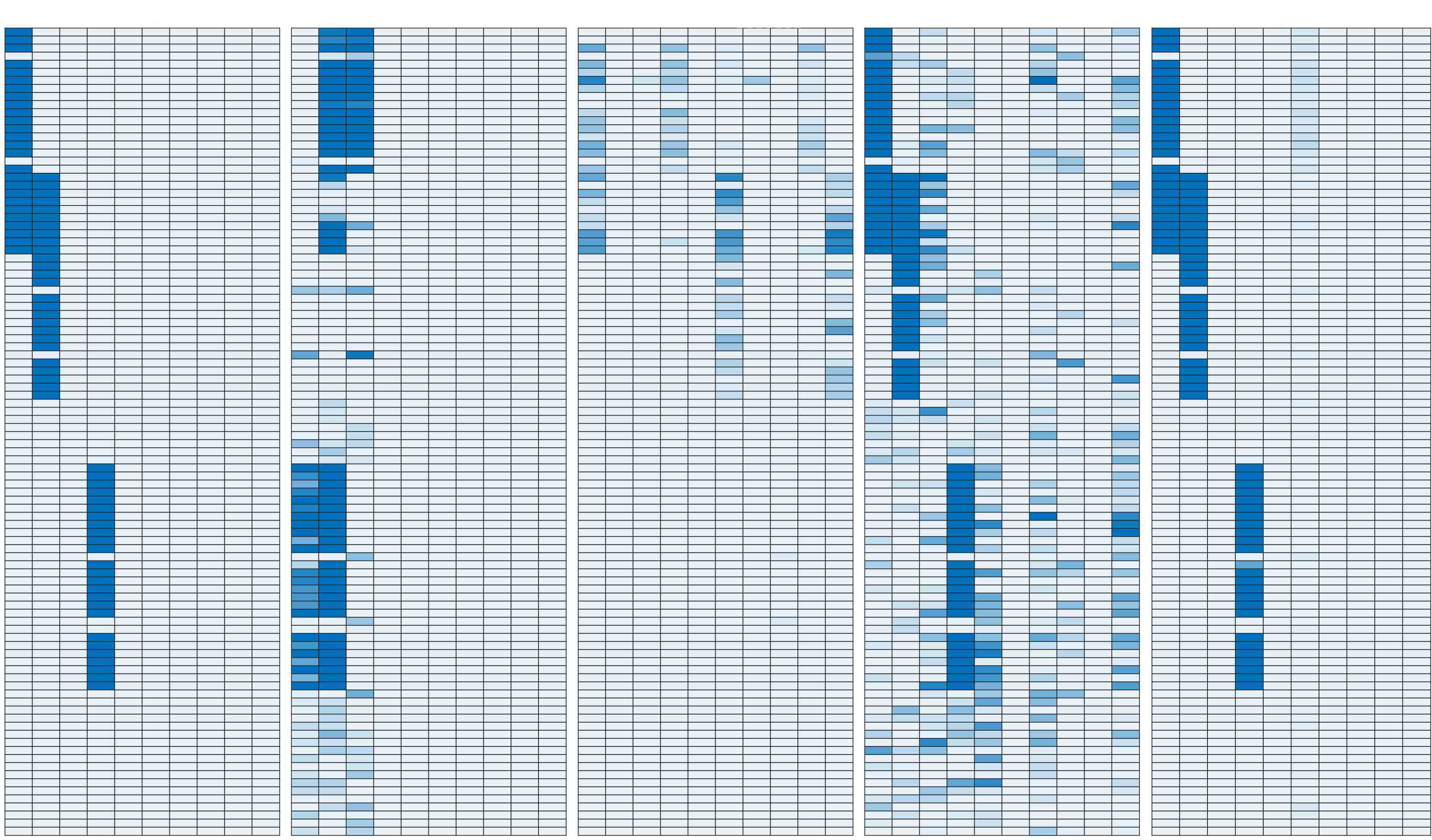}}\label{fig:t300}

\subfigure[$t=400$]{\includegraphics[width=12cm, height=3.5cm]{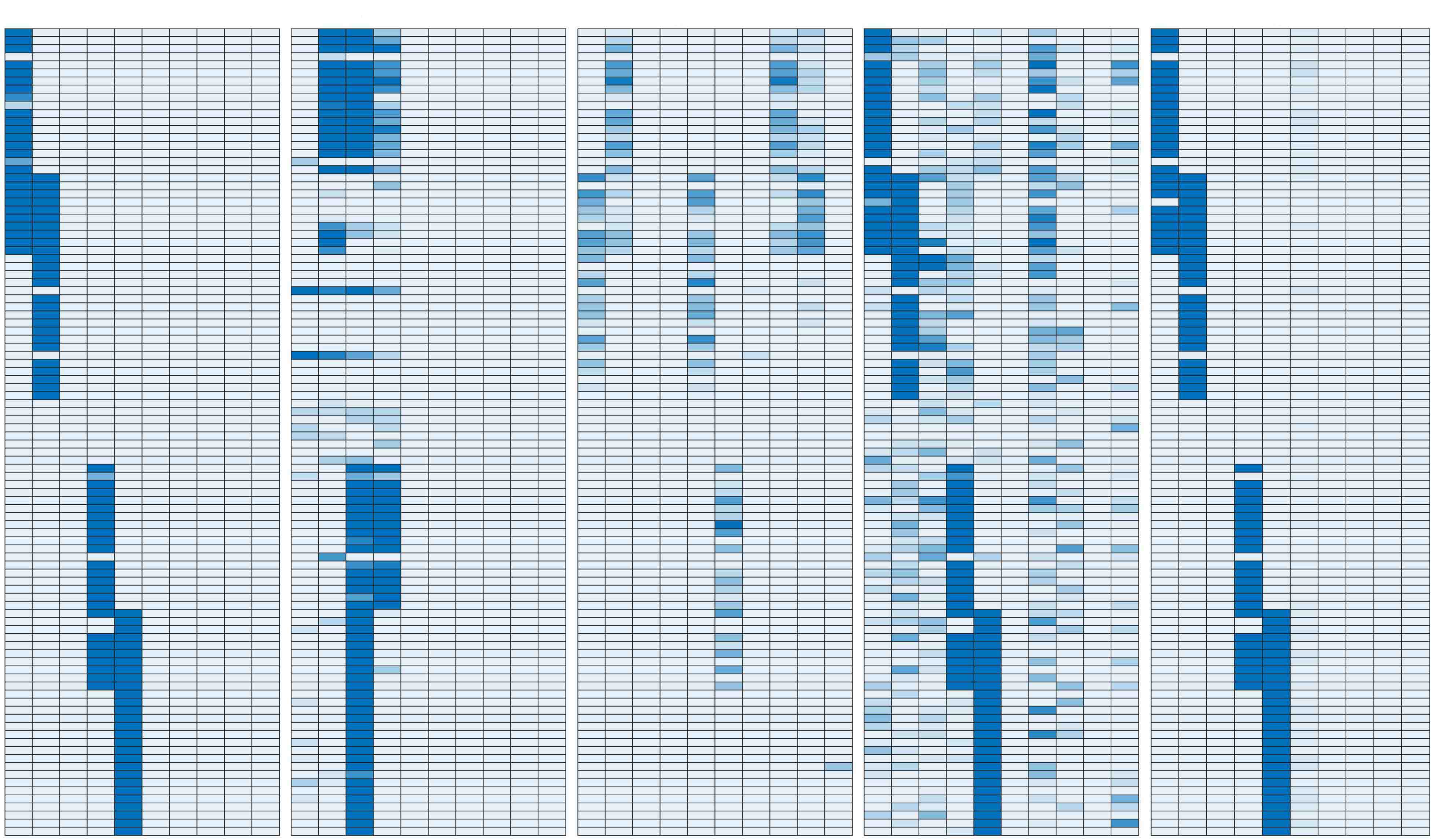}}\label{fig:t400}
\caption{Simulated Example: Heatmaps of true and estimated factor loadings at $t=\{100,200,300,400\}$. Comparisons are made between (from left to right), the true factor loadings, ``Adaptive PCA," ``Sparse PCA" ($K=10$), rolling window spike-and-slab factor analysis ($K=10$), and our dynamic spike-and-slab factor analysis. The first three methods are estimated using a rolling window of $100$ data points.  Factor loadings are absolute and capped at $0.5$ for visibility.}
\label{fig:sim}
 \end{figure}

We compare our proposed dynamic spike-and-slab  factor selection with three other approaches.
The first one is the ``rolling window" version of the static factor analysis with rotations to sparsity by \cite{rovckova2016fast} using $K=10$ (i.e. overshooting the true factor dimensionality).
We compare this approach with  ``Adaptive PCA" of  \cite{bai2002determining}, which corresponds to a rolling-window principal component analysis (PCA) with estimated number of factors, and with ``Sparse PCA"  using $K=10$, which is a rolling-window LASSO-based  regularization method with cross-validation for selecting the level of shrinkage \citep{witten2009penalized}. All these methods are estimated using a rolling window of size $100$, where we generate extra $100$ training data points using the   sparsity pattern $\B^{0}_1$. We choose $\wt\phi=0.95$ and $\sigma^2_\omega=1-\wt\phi^2$.
To deploy the dynamic spike-and-slab priors, we set $\phi_0=0$, $\phi_1=0.98$, $\lambda_0=0.9$, $\lambda_1=10(1-\phi_1^2)$, and $\Theta=0.9$ \citep[following the recommendations in ][]{rockova2018dynamic}.
To improve the performance of our  EM method, we initialize the procedure using the output from the rolling window static spike-and-slab factor model of \cite{rovckova2016fast}.

Focusing on the reconstruction of  factor loadings, we take snapshots at times $t=\{100,200,300,400\}$ and visually compare the output to the truth (Figure~\ref{fig:sim}).
We see that both spike-and-slab methods achieve good recovery. 
However, the static spike-and-slab cannot fully contain the dynamic loadings, where we see a lot of spillover to other factors. 
Dynamic spike-and-slab shrinkage, on the other hand, smooths out the sparsity over time, clearly improving on the recovery.
``Adaptive PCA" performs well, correctly specifying the number of factors. However, the factor loadings are non-sparse and  rotated. ``Sparse PCA" with $K=10$ is  fairly successful, recovering the blocking structure correctly, but splitting the signal among multiple factors
\citep[an observation made also by][]{rovckova2016fast}. 
For the spike-and-slab methods, these patterns can be alternatively obtained by thresholding conditional inclusion probabilities rather than just looking at nonzero entries in $\wh\B_{1:T}$.

\begin{figure}[t!]
\centering
\subfigure[RMSE]{\includegraphics[width=0.45\textwidth]{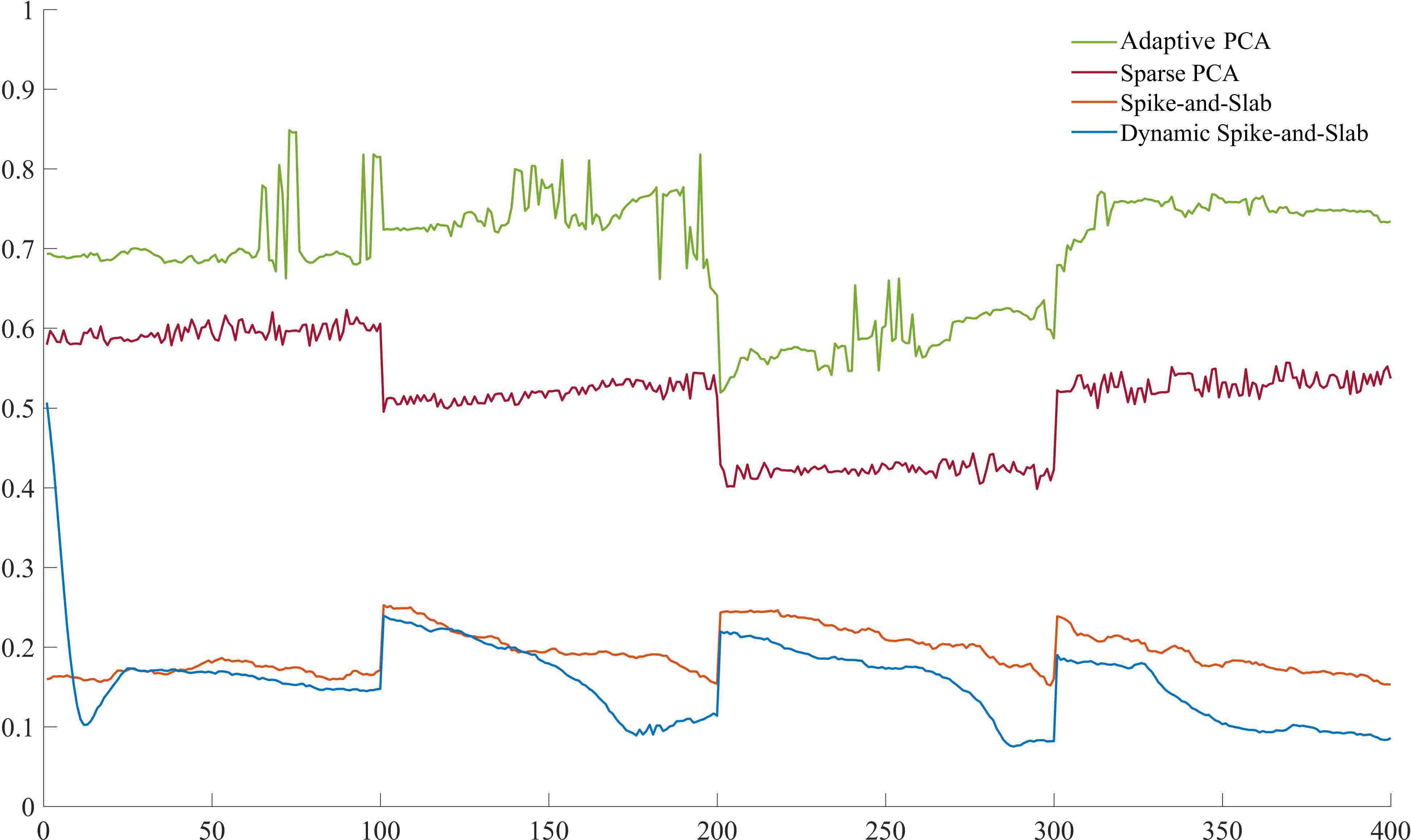}}
\subfigure[$\wh K$]{\includegraphics[width=0.45\textwidth]{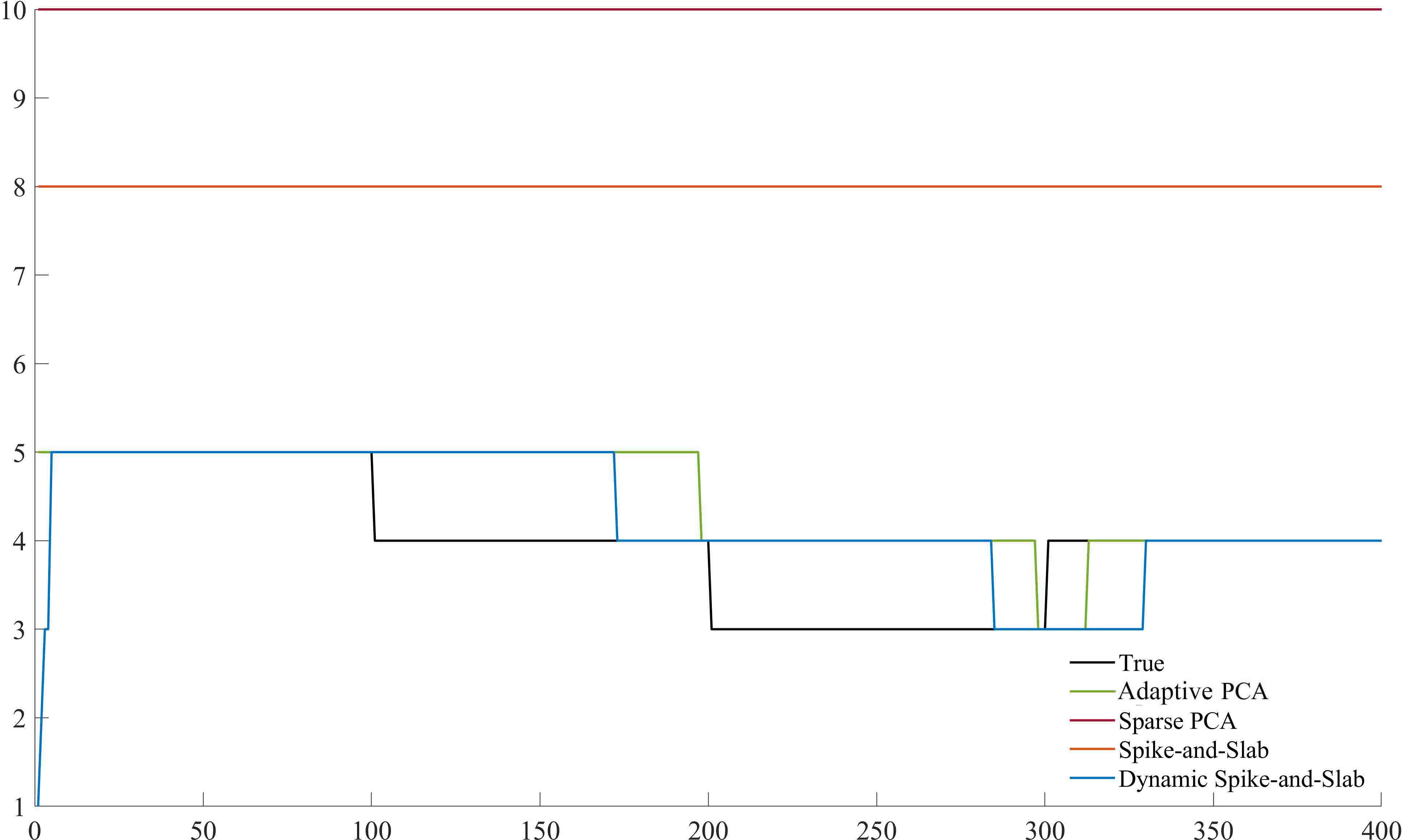}}
\caption{Simulation Study: (Left) The root mean squared error \eqref{rmse} and (Right) the estimated number of factors for ``Adaptive PCA," ``Sparse PCA," static spike-and-slab, and dynamic spike-and-slab, calculated for each $t=1{:}400$. 
\label{fig:rmse}}
\end{figure}

We further explore how the root mean squared errors (RMSE) change over time for one of the simulations (Figure \ref{fig:rmse}).
This is calculated for each $t=1:T$ by
\begin{align}
	RMSE(\wh \B_t) = \sqrt{\frac{tr(\B^0_{t}- \wh\B_t)'(\B^0_{t}- \wh\B_t)}{P\times K}},\label{rmse}
\end{align}
where $\wh\B_{t}$ are the estimated factor loadings at time $t$. Since this comparison is not entirely meaningful due to the rotational invariance, we  compute \eqref{rmse} for the left-ordered variants of these matrices. 
 By looking at the speed of decrease in RMSE after a structural change, it is clear that dynamic spike-and-slab adapts faster compared to its rolling window counterpart.
The drop of RMSE for ``Adaptive PCA" in periods $101{:}200$ and $201{:}300$ can be attributed to the fact that the number of factors was estimated correctly, resulting in  many true zero discoveries. On the other hand, the large estimation error of ``Sparse PCA" is due to the lack of sparsity and scattered structure of the factors. 

Additionally, we plot the estimated number of factors for each method and compare it to the true number of factors.
``Sparse PCA"  overestimates the number of factors (where we regard a factor as active if it has at least one nonzero loading). This indicates that unstructured sparsity is not enough. 
Looking at ``Adaptive PCA" and our dynamic spike-and-slab factor model, we find that both perform similarly well in terms of estimating the number of factors. 
Furthermore, we note that dynamic spike-and-slab adapts faster to factors disappearing, while ``Adaptive PCA" adapts faster to factors reappearing.

We repeat the experiment 10 times and report the average RMSE over each of the four stationary interim time periods in Table~\ref{tab:sim}.
Dynamic spike-and-slab achieves  good recovery, improving upon the rolling window spike-and-slab by as much as 8\% to 34\% (except for the first period). Large recovery errors of the ``Sparse PCA"  method can be explained by factor splitting. While ``Adaptive PCA" does recover the correct number of factors at each snapshot, the loadings are non-sparse, rotated and non-smooth over time.

\begin{table}[t!]
\centering
\scalebox{0.7}{
\begin{tabular}{lrrrrrrrrrrrr}
\hline\hline
                       & \multicolumn{3}{c}{t=1:100}                      & \multicolumn{3}{c}{t=101:200}                    & \multicolumn{3}{c}{t=201:300}                    & \multicolumn{3}{c}{t=301:400}                    \\
                       & \multicolumn{1}{r}{RMSE} & \multicolumn{1}{r}{\%}& \multicolumn{1}{r}{$\wh K$} & \multicolumn{1}{r}{RMSE} & \multicolumn{1}{r}{\%}& \multicolumn{1}{r}{$\wh K$} & \multicolumn{1}{r}{RMSE} & \multicolumn{1}{r}{\%} & \multicolumn{1}{r}{$\wh K$}& \multicolumn{1}{r}{RMSE} & \multicolumn{1}{r}{\%}& \multicolumn{1}{r}{$\wh K$} \\ \hline
Adaptive PCA            & 1.0660             & -266.07     &      5    & 1.0590                  & -400.24               &4.97& 0.9730                   & -250.38     &     3.97     & 1.033                  & -430.01   &   3.88         \\
Sparse PCA             & 0.7862                   & -169.99        &  10      & 0.7260                   & -242.94 &       10       & 0.6377                   & -129.64      &    10      & 0.7383                  & -278.81       &     10   \\
Spike-and-Slab         & 0.1919                   & 34.10          &      8  & 0.2843                    & -34.29                 &8& 0.2988                    & -7.60           &   8   & 0.2447                    & -25.60      &    8       \\
Dynamic Spike-and-Slab & 0.2912                   & -             &   4.89      & 0.2117                   & -                      &4.72& 0.2777                    & -            &   3.84       & 0.1949                  & -              &     3.71   \\ \hline\hline
\end{tabular}}
\caption{\small Simulation Study:  Performance evaluation of the latent factor methods compared to the true coefficients for $t=1{:}400$. Performance is evaluated based on RMSE within each evaluation period. \% is the performance gain compared to dynamic spike-and-slab. $\wh K$ is the average number of factors estimated during that period.
\label{tab:sim}}
\end{table}


\section{Empirical Study \label{sec:study} }
The empirical application concerns a large-scale monthly U.S. macroeconomic  database,  
comprising a balanced panel of $P=127$ monthly macroeconomic and financial variables tracked over the period of $2001/01$ to $2015/12$ ($T=180$).
These variables are classified into eight main categories, depending on their economic meaning: {\sl Output and Income, Labor Market, Consumption and Orders, Orders and Inventories, Money and Credit, Interest Rate and Exchange Rates, Prices}, and {\sl Stock Market}.  
A detailed description of how variables  were  collected and constructed is provided in \cite{mccracken2016fred}. 
A quick table of names and groups of each variable is in  the Appendix (Table \ref{tab:macrovar}).
The variables were centered to have mean zero and standardized following the procedures in \cite{mccracken2016fred}.

The purpose of conducting  a sparse latent factor analysis on a large-scale economic dataset, such as this one, is at least twofold. 
Due to the group structure of the data, it is natural to assume that the measured indicators are tied via a few  latent factors, the basic premise of latent factor modeling.
Moreover, 
we expect the sparse latent structure to pickup clusters of dependence structures that capture the interconnectivity of indicators spanning many {\sl different} aspects of the economy.
Sparsity will help extract such interpretable structures.
Second, given the dynamic nature of the economy, there is a substantial interest  in understanding  how these dependencies change over time and-- in particular-- how they are affected by  shocks.
We anticipate non-negligible shifts in the economy, as the data spans over  the housing bubble deflation after 2006 and the great financial crisis in late 2008, which led to the Great Recession.
Understanding the interplay between  contributing factors to the financial crisis has been a subject of rigorous research \citep[see for example,][]{financial2011financial,reinhart20082007,mian2009consequences,mian2011house,mian2013household,chodorow2014effects,benmelech2017real}. 
Our analysis is purely data-driven and thereby descriptive rather than causally conclusive. We attempt to characterize patterns of shock proliferation and permanence of structural changes  of the economy using our dynamic factor model.

As the dataset is considerably richer than our simulated example, we expand the  model \eqref{eq:lf1} by incorporating a dynamic intercept to capture location shifts   that could not be easily standardized away.
 {The intercepts $c_{jt}$ follow independent  random walk evolutions with an initial condition $c_0\sim N(0,1)$. The  initial condition for the SV variances is  $1/\sigma_{j0}^2\ind G(n_0/2,d_0/2)$ for $1\leq j\leq P$ with $n_0=20$ and $d_0=0.002$}. The discount factor is set to 0.95.




\begin{figure}[t!]
     \subfigure[``Adaptive PCA"]{
   \begin{minipage}[b]{0.3\textwidth}
       \centering \includegraphics[width=3cm, height=5cm]{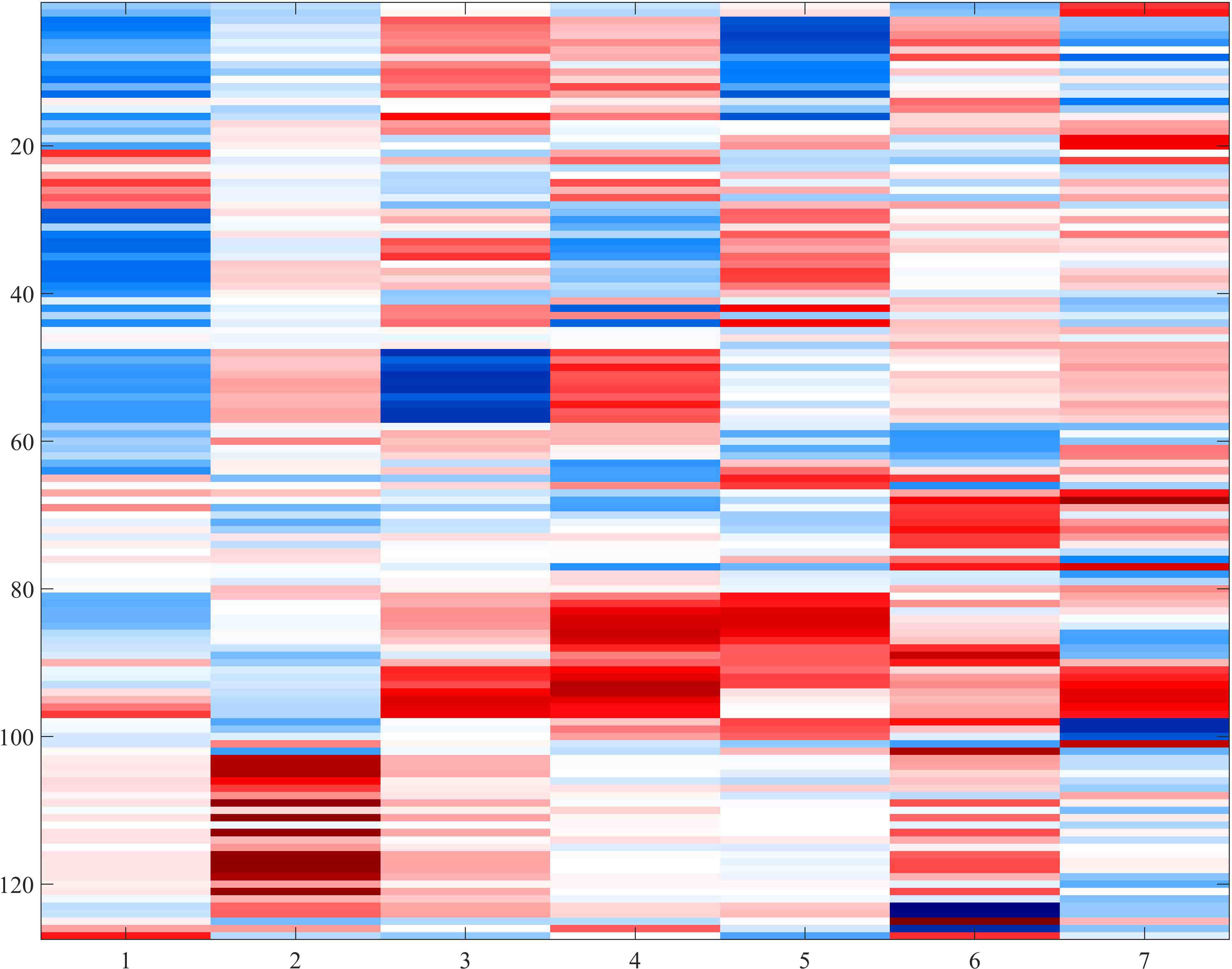}
    \end{minipage}}
     \subfigure["Sparse PCA"]{
    \label{fig1:dyna2}
    \begin{minipage}[b]{0.3\textwidth}
       \centering\includegraphics[width=3cm, height=5cm]{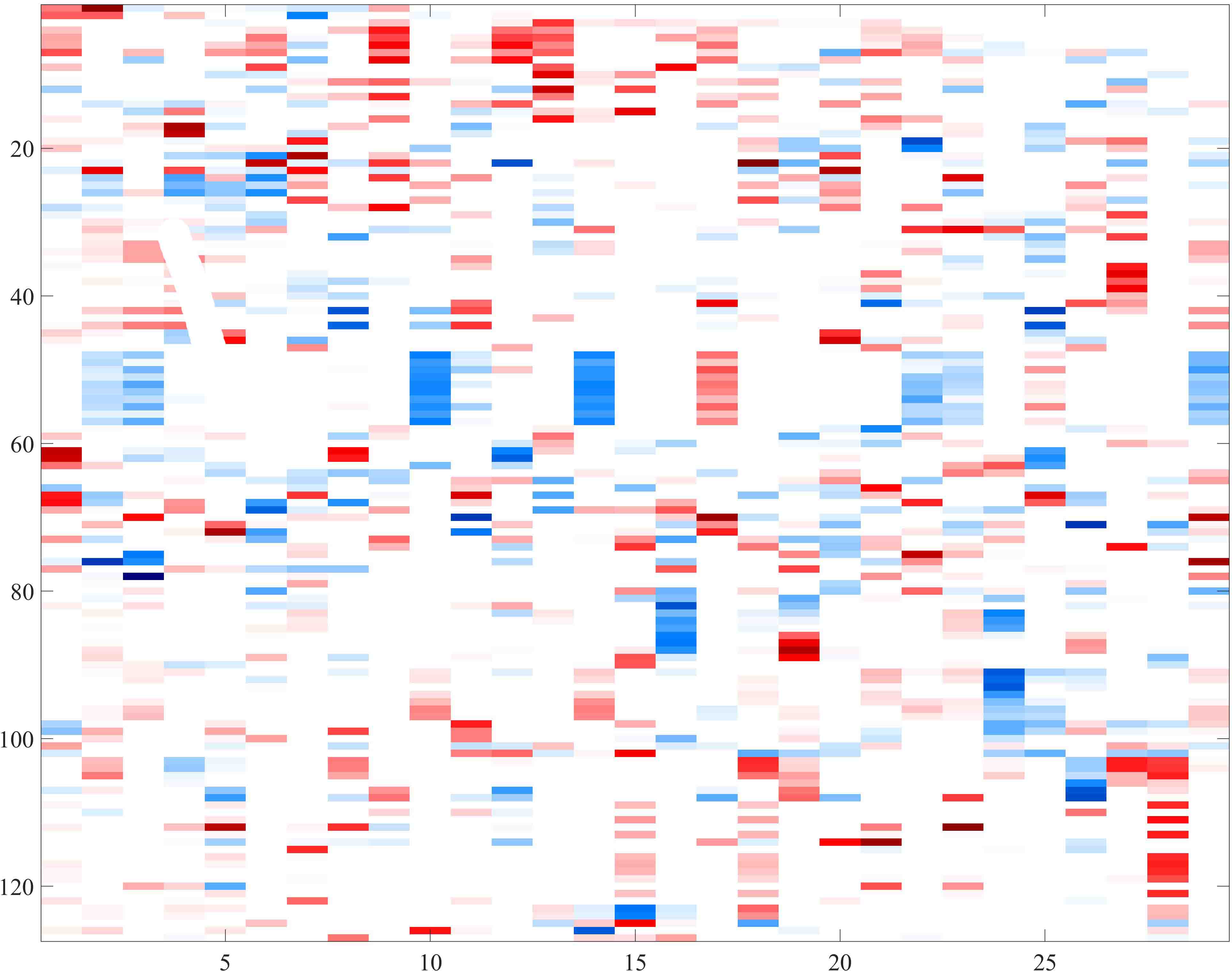}
    \end{minipage}}
    \subfigure["Sparse PCA"]{
    \label{fig1:dyna2}
    \begin{minipage}[b]{0.3\textwidth}
       \centering\includegraphics[width=3cm, height=5cm]{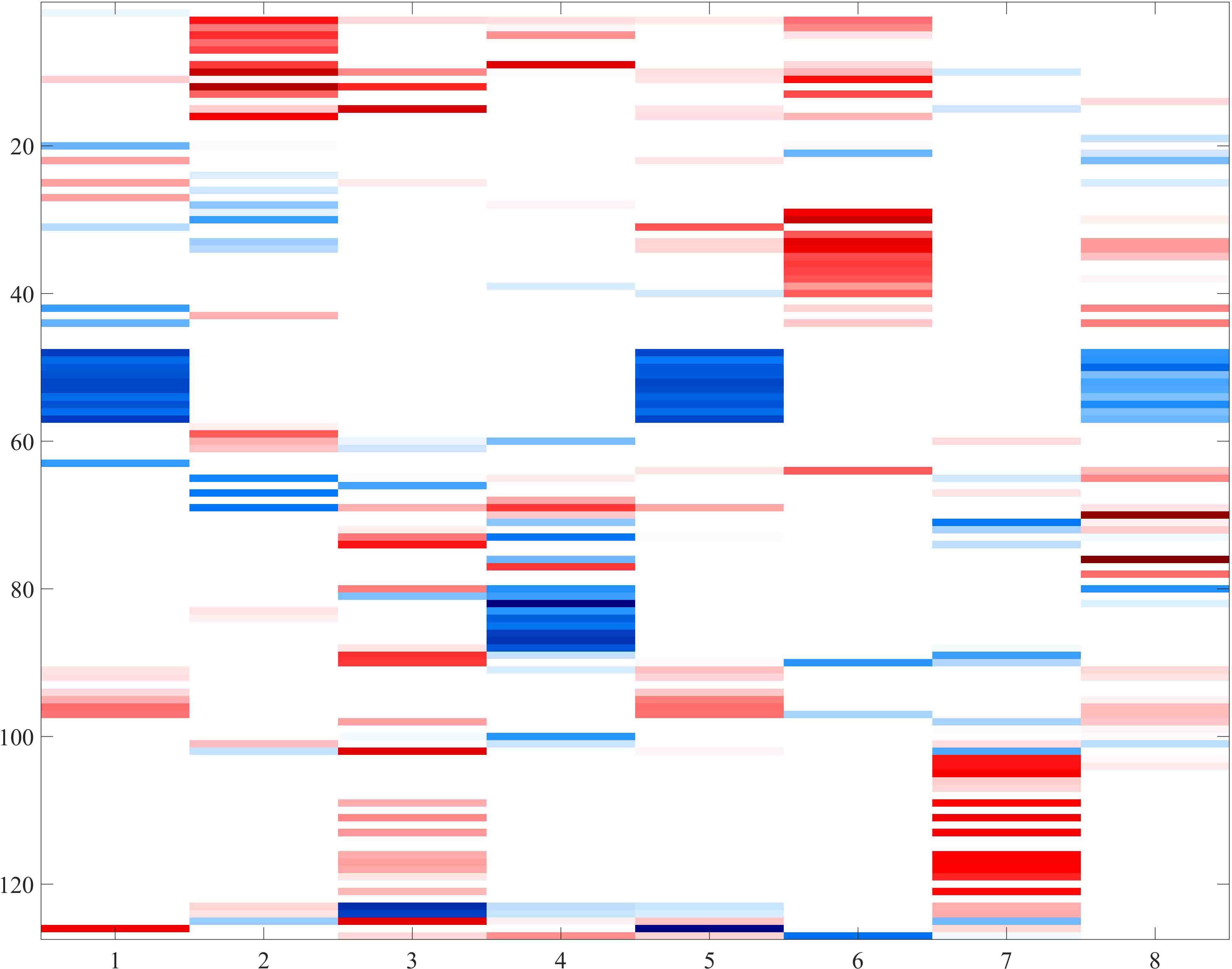}
    \end{minipage}}
\caption{Macroeconomic Study:  Estimated factor loadings using ``Adaptive PCA" (Left), ``Sparse PCA" with number of factors set as 30 (Middle), and ``Sparse PCA" with number of factors set to 8 from the results of ``Adaptive PCA"  (Right) at $t=2015/12$, with the number of series on the y-axis and the number of factors in the x-axis. The factor loading are estimated using a $10$ year rolling window.
\label{fig:PCA1512}}
\end{figure}


First, we examine one snapshot of the output from ``Adaptive PCA" and ``Sparse PCA" (described in Section \ref{sec:simstudy}) at time $2015/12$ (Figures~\ref{fig:PCA1512}).
Both  methods do pick up  certain groupings, but do not yield interpretable enough representations. This  is likely due to overestimation of the number of factors (Figure \ref{fig:PCA1512} (b)), factor rotation and lack of sparsity (Figure \ref{fig:PCA1512} (a)) and/or factor splitting (Figure \ref{fig:PCA1512} (c)). 
Next, we deploy the rolling window  spike-and-slab factor method with a training period {of 10 years} to obtain starting values for our dynamic factor model.
Priors and their hyper-parameters were  chosen as in  the simulation study.  We choose a generous upper bound $K=126$ on the number of factors, letting the sparsity rule out factors that are irrelevant.

We now examine   the output of our procedure at three time points: 2003/12, 2008/10, and 2015/12.
These three snapshots are of particular interest as they represent three distinct  states of the economy: relative stability (2003), sharp economic crisis (2008), and recovery (2015).
2008/10 is at the onset of the great financial crisis, where deflation of the housing bubble after 2006 lead to mortgage delinquencies and financial fragility \citep{financial2011financial}.
This distress permeated throughout the rest of the economy, including the labor market,  leading to the deepest recession in post-war history. 

\begin{figure}[t!]
\centering
\includegraphics[width=0.8\textwidth]{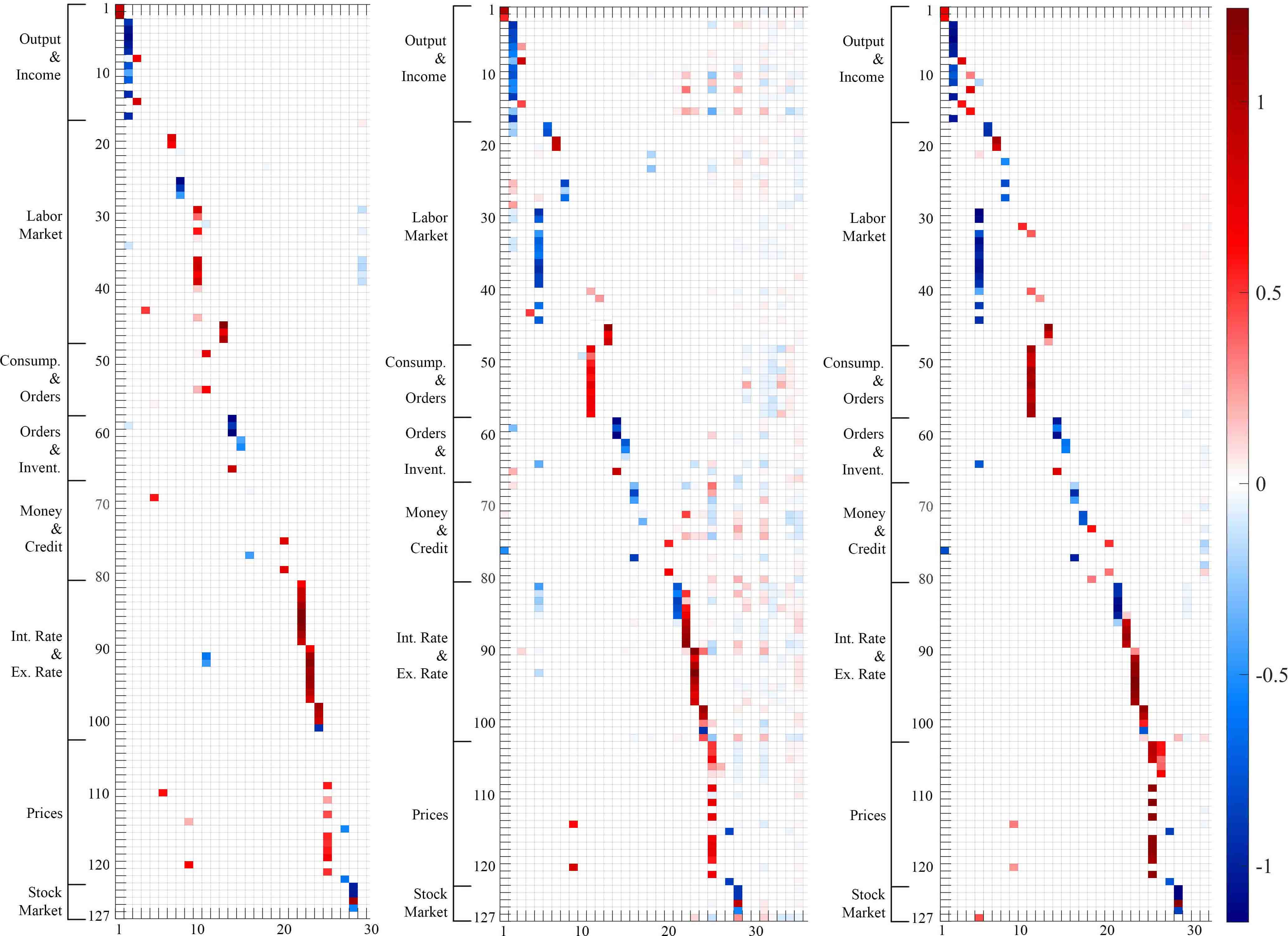}
\caption{Macroeconomic Study:  Estimated factor loadings using dynamic sparse factor analysis at $t=2003/12$ (left), $t=2008/10$ (center), $t=2015/12$ (right), with the original series on the y-axis and the  factors in the x-axis.  The factor loading are estimated dynamically over the period $2001/1{:}2015/12$.}\label{fig:macroDSScol}
\end{figure}

The heatmap of estimated factor loadings at time $2003/12$ is in Figure~\ref{fig:macroDSScol} (left). The output has been left-ordered based on the results at $2015/12$, where the more active factors are on the left, in the order of data series, and some of the less active right-most factors  (with small or zero loadings) are omitted.
There are $24$ active factors in total (i.e. factors with at least two non-negligible non-zero factor loadings), with only $5$ factors that cluster eight or more series (Factors 2, 10, 22, 23, and 25).
Since the variables are grouped by their economic meaning, this type of clustering is not entirely unexpected. 
For example, Factor 2 includes CMRMTSPLx (real manufacturing and trade industry sales), all industrial production indices except nondurable materials, residential utilities, and fuels,  CUMFNS (capacity utilization), DMANEMP (durable goods employment), and ISRATIOx (manufacturing and trade inventories to sales ratio). 
This factor could be interpreted as a factor for {\sl durable goods}, which include industries that are more susceptible to economic trends, where sales, inventories, industrial production, capacity utilization, and employment are all connected.
Conversely, we expect nondurable goods, such as utilities and fuels, to have a different dynamic than durable goods, which is reflected in the exclusion of those indices in Factor 2.
Similarly, Factor 10 includes employment data (except for mining and logging, manufacturing, durable goods, nondurable goods, and government), Factor 22 includes interests rates (fed funds rate, treasury bills, and bond yields), Factor 23 includes the spread between interest rates minus fed funds rate, and Factor 25 includes consumer price indices except apparel, medical care, durables, and services, as well as personal consumptions expenditures on nondurable goods.
All of these factors produce meaningful and mostly separated clusters  that largely conform with economic intuition.

During the crisis (Figures~\ref{fig:macroDSScol}; center), radical changes occur in the factor structure.  
Concerning Factor 2, the dependence structure expands, now spanning over nondurables and fuels, as well as  HWI (the help wanted index), UNEMP15OV (unemployment for 15 weeks and over), CLAIMSx (unemployment insurance claims), and  PAYEMS (employment, total non-farm, goods-producing, manufacturing, and durable goods).
This indicates that the shock might have affected relatively stable industries and unemployment, with the co-movement across industries being  largely synchronized under distress (with the exception of residential utilities). 
Another interesting observation is the emergence of new factors.
In particular, Factor 11, which includes housing starts and new housing permits in different regions in the U.S., was {\sl not} present pre-crisis and now surfaces as a connecting thread between housing markets across regions.
While in $2003/12$ the latent factors were largely separated (loadings had little overlap), we now see  at least two factors (namely Factor 25 and 28),
  whose loadings are non-sparse and far-reaching. 
In particular, Factor 28 emerges as a non-sparse link between many different sectors of the economy,
 including retail sales, industrial production, employment (in particular financial services), real M2 money stock, loans, BAA bond yields (but not AAA), exchange rates, consumer sentiment, investment and, most importantly, the stock market indices, including the S\&P 500 and the VIX (i.e. the fear index). 
Factor 25, on the other hand, is driven mainly by prices (e.g. CPI). Both of these factors 
    could be potentially interpreted as crisis factors
as they are connected to the various corners of the economy, except  Consumption and Orders; the housing market.
The ``orthogonality" between the housing market factor (Factor 11) and the ``crisis factors" (Factor 25 and 28) may suggest that, while the crisis was triggered by the housing market, the main catalyst of the recession was the financial market. While our analysis does not necessarily prove this hypothesis, it aligns with previous lines of reasoning. In particular, there have been  arguments that the devaluation of securities, including mortgage backed securities, ultimately led to curtailed lending and decreased investment and  consumption \citep{chodorow2014effects,benmelech2017real}.


Finally, Figure~\ref{fig:macroDSScol} (right)  shows the end of the analysis at $2015/12$, where the economy has mostly recovered from the Great Recession, but has fundamentally changed from what it was before.
Although most of the factor overlap has dissipated, we see a notably different structure compared to 2003.
In particular, Factor 5 (employment) and Factor 11 (housing) persevere from the crisis. Moreover, the ``crisis factors" Factor 25 and 28, representing the  prices and the stock market,
are no longer strongly tied to other parts of the economy (labor, output, interest and exchange rates, etc.).
Factor 2 is one of the few factors that have returned back to its original structure, except for CMRMTSPLx and industrial production of nondurable consumer goods.
Its dependence with the labor market (e.g. unemployment) has  disappeared, suggesting that industry production is no longer in co-movement with the labor market.

\begin{figure}[t!]
     \subfigure{
   \begin{minipage}[b]{0.5\textwidth}
       \centering \includegraphics[width=1\textwidth]{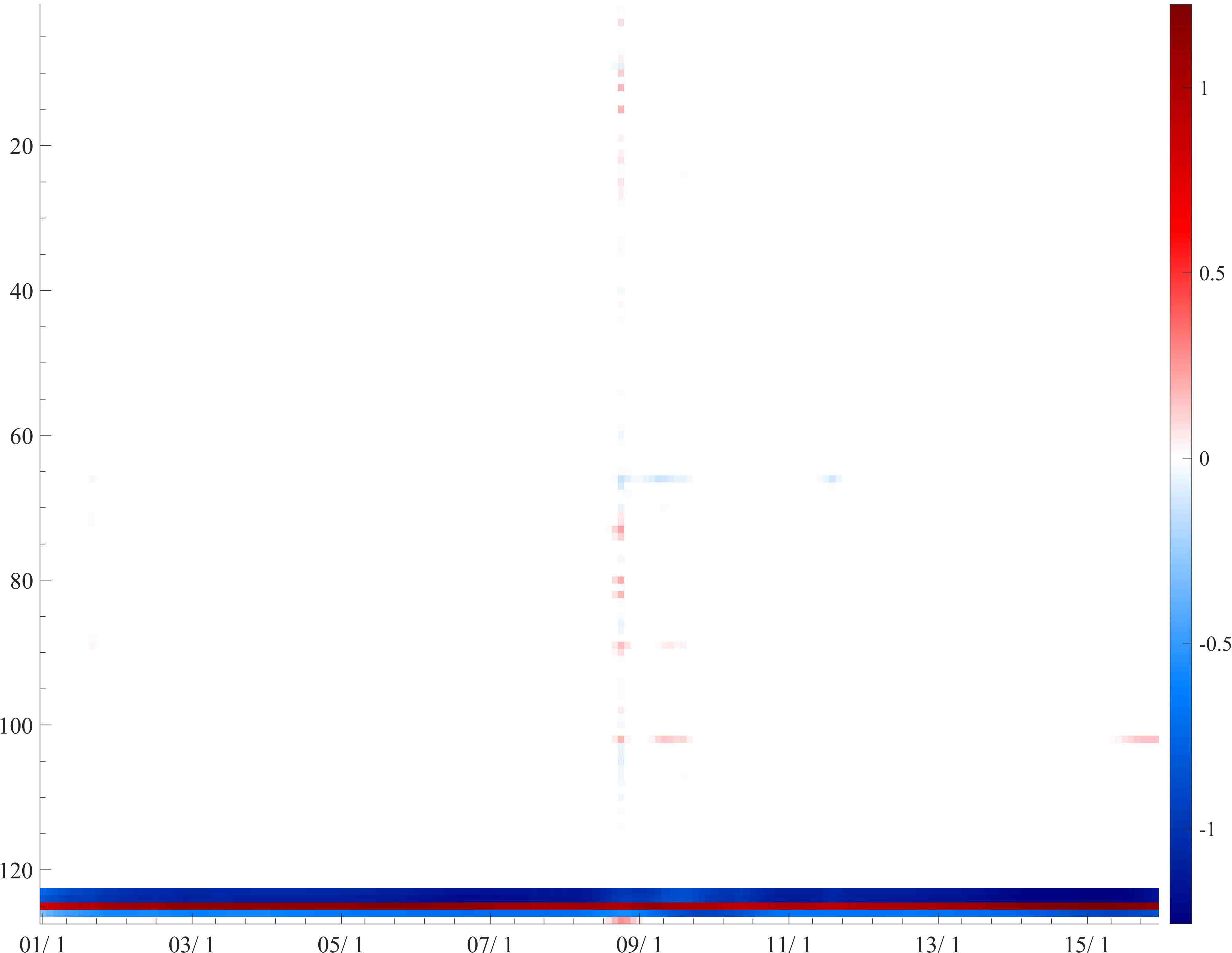}
    \end{minipage}}
     \subfigure{
    \label{fig1:dyna2}
    \begin{minipage}[b]{0.5\textwidth}
       \centering\includegraphics[width=1\textwidth]{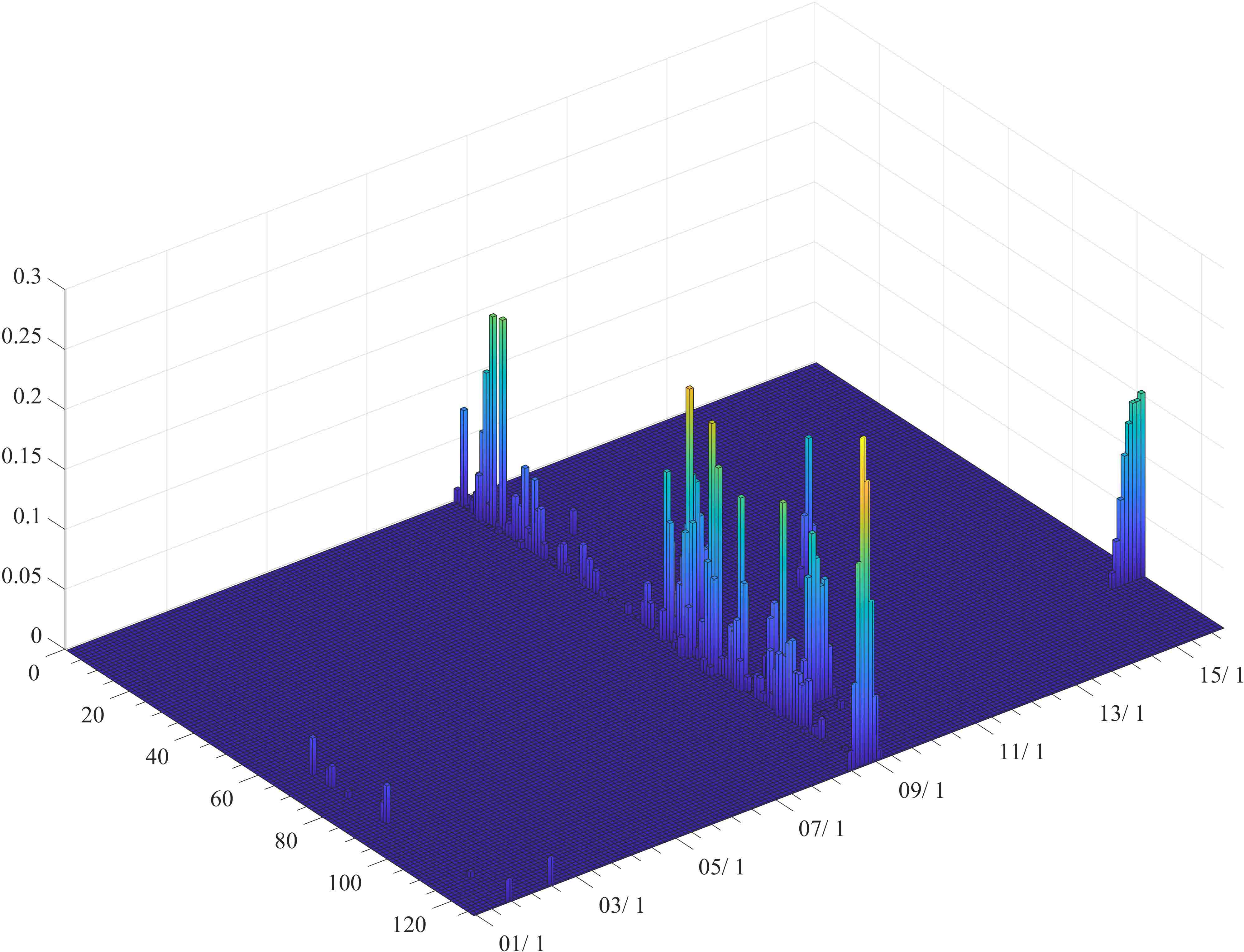}
    \end{minipage}}
\caption{Macroeconomic Study:  Estimated factor loadings for Factor 28 using dynamic spike-and-slab from $t=200	/12{:}2015/12$, with a heatmap of the entire factor loadings (Left) and a 3-D plot of the factor loadings with the loadings on 123-126 (S\&P related indices) set to zero to increase visibility.}
\label{fig:28loading}
\end{figure}

We also obtain  insights into the effects and  duration of the crisis by looking at the evolution of the factor loadings for one of the ``crisis" factors, Factor 28. Figure~\ref{fig:28loading} shows a dynamic heatmap and a $3$-D plot of $\beta_{jk}^t$  for $1\leq j\leq 127$ (y-axis) and $1\leq t\leq 180$ (x-axis) with $k=28$.
For the $3$-D plot, the loadings on the S\&P indices are suppressed to zero in order to improve visibility.
The figure reveals a spur of activity around the sharp financial crisis (late 2008 and early 2009), where the contagion  battered  multiple corners of the economy.
The duration  of the active loadings provide additional insights.
For example, the loadings on VIX {(series 127)} emerges and disappears in a eight month span {from 06/2008 to 02/2009}, while the loadings on the exchange rate between U.S. and Canada lasts for $17$ months. However, most factor loadings seem to only emerge for about 4-6 months.

To understand the degree of connectivity/overlap between factors, we plot the average number of active factors per series over time (Figure~\ref{fig:n_fac}).
More overlap indicates a more intertwined economy.   We observe an increase in late $2008$, reflecting the emergence pervasive crisis factor(s), {\color{black} as well as its build up from mid-2006.
Another point to note is that the level pre-crisis is comparatively lower than post-crisis, indicating a structural shift is the economy brought on by the crisis.}

\begin{figure}[t!]
\centering
\includegraphics[width=0.8\textwidth]{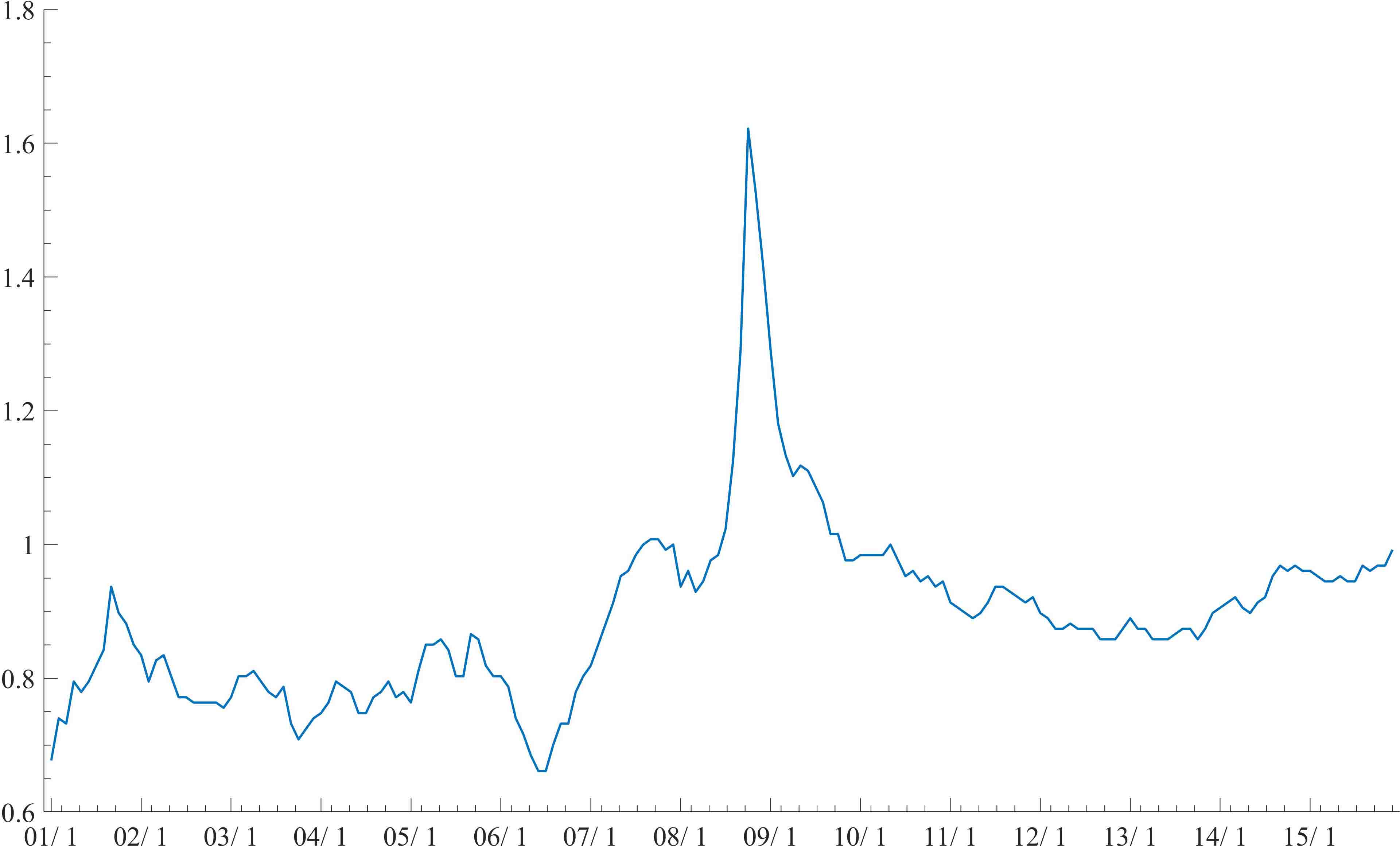}
\caption{Macroeconomic Study:  The average number of estimated active factors  (with absolute loadings above $0.1$) per series over the period $2001/1{:}2015/12$. 
\label{fig:n_fac}}
\end{figure}

We further our analysis with a few insights into the idiosyncratic variances for  variables related to the housing market:
 HOUST (total housing starts) and its regional variants (North East, Mid-West, South, and West).
Housing starts is  the seasonally adjusted number of new residential construction projects that have begun during any particular month and, as such,  is a key part of the U.S. economy, which relates to employment  and many industry sectors including banking (the mortgage sector), raw materials production, construction, manufacturing, and real estate.
In our earlier analysis (Figure~\ref{fig:macroDSScol}) we found that, while regional indicators were not clustered  pre-crisis, persistent clustering occurs post-crisis.
Figure~\ref{fig:idrisk} portrays the series of residual uncertainties $\{\sigma^2_{jt}:1\leq t\leq T\}$ for each  regional housing starts indicator.
We find several interesting patterns. Figure~\ref{fig:idrisk} indicates that increased uncertainty in housing starts  is a global phenomenon but that there is heterogeneity across regions as to the magnitude and timing. 
{\color{black} For example, we find that the West region to react the earliest, followed by Mid-West and South.
North-East is somewhat of an exception, as the idiosyncratic variance starts out greater than the other series, falling off pre-crisis, increasing during the crisis, and tapering off to a level similar to the other regions.
The speed of mounting uncertainty could be associated with the deflation of the housing bubble after 2006 \citep{financial2011financial}.
As the economy recovers from the Great Recession, we observe a gradual decrease in uncertainty, where different regions recover at different paces.}

\begin{figure}[t!]
\centering
\includegraphics[width=0.8\textwidth]{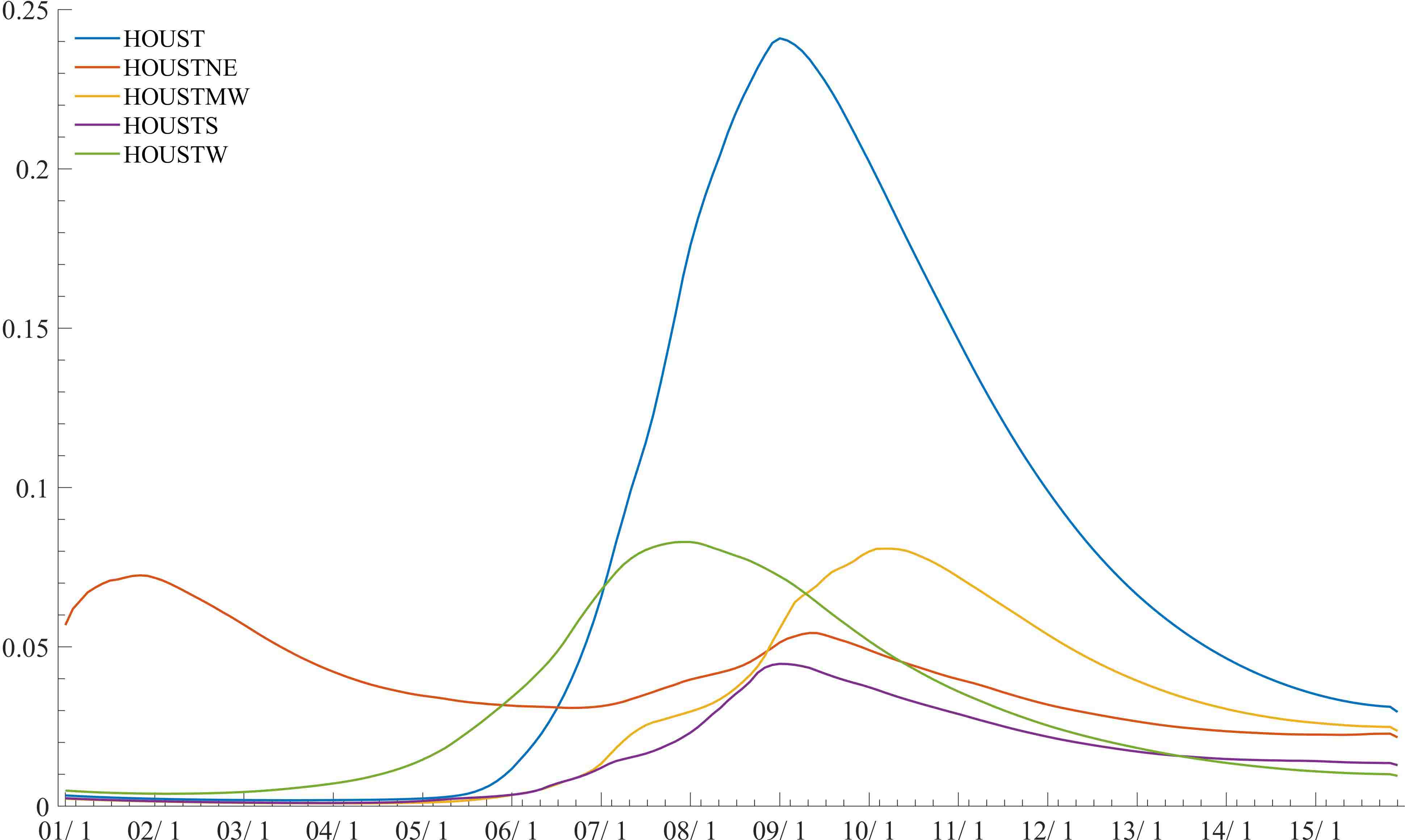}
\caption{Macroeconomic Study:  The idiosyncratic variance, $\bSigma_t$, of U.S. housing starts, over the period $2001/1{:}2015/12$. 
\label{fig:idrisk}}
\end{figure}

\section{Further Comments \label{sec:conclusion} }
Motivated by a topical macroeconomic dataset, we developed  a Bayesian method for dynamic sparse factor analysis for large-scale time series data.
Our proposed methodology aims to tackle three challenges of dynamic factor analysis: time-varying patterns of sparsity, unknown number of factors, and identifiability constraints. By deploying dynamic sparsity, we successfully recover interpretable latent structures that automatically select the number of factors and that incorporate time-varying loadings/factors. We successfully applied our  methodology on a nontrivial simulated example as well as a real dataset comprising of 127 U.S. macroeconomic indices tracked over the period of the Great Recession (and beyond) and obtained several interpretable findings.

Our methodology can be enriched/extended in many ways.
One possible extension would be to develop a latent variable method that can capture  within, as well as  between, connectivity of several high-dimensional time series.
This could be achieved with  a dynamic extension of sparse canonical correlation analysis \citep{witten2009penalized}.
Our method can also be embedded within FAVAR models \citep{bernanke2005measuring} that include both observed and unobserved  predictors.
{Additionally, throughout our analysis we have assumed the covariance of the latent factors to be fixed over time and equal to an identity matrix, one could in principle incorporate dynamic variances with stochastic volatility modeling. }

One possible shortcoming of our proposed methodology, which is shared by all EM based estimation strategies, is the lack of uncertainty assessment, which is essential for forecasting.
The EM algorithm, however, was the key to obtaining  interpretable latent structures.
To achieve both,  one could impose  identification constraints, such as \cite{nakajima2012dynamic,nakajima2013bayesian},
and perform MCMC for DSS priors along the lines of  \cite{rockova2018dynamic}.
Another approach would be to apply our method  simply as a means of  obtaining  identifiability constraints (i.e. the sparsity pattern) and then reestimate the nonzero loadings with an MCMC strategy.
While this would not quantify any sparsity-selection uncertainty, it  would be an effective way to balance interpretability and forecasting/decision making.
Another unavoidable feature of our method is its sensitivity to starting values. We strongly recommend using the output from the rolling window spike-and-slab factor model. 


\newpage
\bibliographystyle{elsarticle-harv}
\bibliography{References}

\newpage

\clearpage
\setcounter{page}{1}
\begin{center}
{\Large Dynamic Sparse Factor Analysis} 

\bigskip
{\Large  Supplementary Material} 

\bigskip\bigskip
\end{center}

\appendix

\section{Appendix}

\subsection{Derivation of the E-step}\label{sec:estep}
We now outline the steps of the parameter expanded EM algorithm.
In the E-step, we compute the conditional expectation of the augmented and expanded log-posterior with respect to the missing data $\BOmega$ and $\BGamma$, given observed data $\Y$ and  the parameter values $\BD^{(m)}$ obtained at the previous M-step setting $\A_t=\sigma^2_\omega\bm I_K$. We can write

\begin{align}\label{eq:expect}
\E_{\BGamma, \BOmega \C \Y, \Delta^{(m)}} [\log \pi (\B^\star_{0:T},\bSigma_{1:T}, \A_{1:T}, \BGamma, \BOmega \C\Y)] =& Q_1(\B_{0:T}^\star \C \bSigma_{1:T}) + Q_2(\bSigma_{1:T})  + Q_3(\A_{1:T}) + C.
\end{align}
Define  
 $\bomega_{t|T}=\E_{\BOmega}[\bomega_t\C\Y,\bm\Delta^{(m)}]$,  $\V_{t|T}=\mathrm{cov}[\bomega_t\C\Y,\BD^{(m)}]$. 
 The terms $\bomega_{t|T}$ and $\V_{t|T}$ represent the best linear estimator for $\bomega_t$ using all observations and the corresponding covariance matrix, respectively. With $\V_{t,t-1|T}=\mathrm{cov}[\bomega_t,\bomega_{t-1}\C\Y,\BD^{(m)}]$ we denote the covariance matrix of $\bomega_t$ and $\bomega_{t-1}$ given the data $\Y$ and $\BD^{(m)}$. These quantities can be obtained from the Kalman Filter and Smoother Algorithm (Table \ref{kalman}).

  The functions $Q_1(\cdot)$, $Q_2(\cdot)$ and $Q_3(\cdot)$ in \eqref{eq:expect} can be written as follows:
\begin{align*}
- Q_1(\B_{0:T}^\star \C \bSigma_{1:T})  =& C+\frac{1}{2}\sum_{t=1}^T \sum_{j=1}^P \log \sigma_{jt}^2 \\
&+ tr\left\{\frac{1}{2} \sum_{t=1}^T \bm{\Sigma}_t^{-1} \left[(\Y_t-\B_t^\star \bomega_{t|T})(\Y_t-\B_t^\star \bomega_{t|T})'+\B_t^\star \V_{t|T} \B^{\star'}_t\right]\right\} \\
& + \sum_{j=1}^P \sum_{k=1}^K \left[ \frac{\langle \gamma_{jk}^0 \rangle (\beta_{jk}^{0*})^2}{2\lambda_1/(1-\phi^2)}+ (1-\langle \gamma_{jk}^0 \rangle) | \beta_{jk}^{0*} | \lambda_{0} \right]\\
& + \sum_{t=1}^T \sum_{j=1}^P \sum_{k=1}^K \left[\frac{ \langle \gamma_{jk}^t \rangle (\beta_{jk}^{t*}-\phi \beta_{jk}^{t-1*})^2}{2\lambda_1}
+(1-\langle \gamma_{jk}^t \rangle) | \beta_{jk}^{t*} | \lambda_0 \right],
\end{align*}
where
\begin{align*}
\langle \gamma_{jk}^0 \rangle & = \frac{\Theta \psi_1(\beta_{jk}^{0}|0,\frac{\lambda_1}{1-\phi^2})}{\Theta \psi_1(\beta_{jk}^{0}|0,\frac{\lambda_1}{1-\phi^2})+ (1-\Theta) \psi_0(\beta_{jk}^{0}|0,\lambda_0)}, \\ 
\langle \gamma_{jk}^t \rangle & = \frac{\theta_{jk}^t \psi_1(\beta_{jk}^{t}|\phi \beta_{jk}^{t-1},\lambda_1)}{\theta_{jk}^t \psi_1(\beta_{jk}^{t}|\phi \beta_{jk}^{t-1},\lambda_1)+ (1-\theta_{jk}^t) \psi_0(\beta_{jk}^{t}|0,\lambda_0)},
\end{align*}
{\begin{align*}
-Q_2(\bSigma_{1:T})& =\sum_{t=1}^{T-1} \sum_{j=1}^P \left [pen(\sigma_{jt}^2 \C \sigma_{j(t-1)}^2) + pen(\sigma_{jt}^2 \C \sigma_{j(t+1)}^2) \right ] + \sum_{j=1}^P pen(\sigma_{jT}^2 \C \sigma_{j(T-1)}^2)
\end{align*}
where 
\begin{align*}
    pen(\sigma_{jt}^2 \C \sigma_{j(t-1)}^2) & = \left(\frac{\delta n_{t-1}}{2}-1\right)\log{\sigma_{jt}^2}- \left(\frac{(1-\delta)n_{t-1}}{2}-1\right) \log {\left(1-\frac{\delta \sigma_{j(t-1)}^2}{\sigma_{jt}^2}\right)}, \\
    pen(\sigma_{jt}^2 \C \sigma_{j(t+1)}^2) & = -\left(\frac{\delta n_{t}}{2}-1\right)\log{\sigma_{jt}^2} + \left(\frac{(1-\delta)n_{t}}{2}-1\right) \log {\left(1-\frac{\delta \sigma_{jt}^2}{\sigma_{j(t+1)}^2}\right)}, \\
\end{align*}
and

{
\spacingset{1}
\begin{table}[t!]
\small
\begin{center}
\scalebox{0.8}{
\begin{tabular}{|l|l|}
\hline
\multicolumn{2}{|c|}{\cellcolor[HTML]{C0C0C0} \textbf{Algorithm:} \textit{ Kalman Filter and Smoother}} \\ \hline \hline
\multicolumn{2}{|c|}{Initialize $\bomega_{0 \C 0}=\bm 0$ and $\V_{0 \C 0}= \sigma^2_\omega/(1-\wt\phi^2)\bm I_K$}                                                    \\
\multicolumn{2}{|c|}{Repeat the Prediction Step and Correction Step for $t=1,\dots,T$}         \\ \hline
\cellcolor[HTML]{C0C0C0}Prediction Step                   &   $\bomega_{t \C t-1}=\bomega_{t-1 \C t-1}$                 \\
                                                          &    $\V_{t \C t-1} = \V_{t-1 \C t-1}+ \sigma^2_\omega\bm I_{K}$                   \\ \hline
\cellcolor[HTML]{C0C0C0}Correction Step                   &  $\bm{K}_t=\V_{t \C t-1}\B_t'(\B_t\V_{t \C t-1}\B_t'+\bSigma_t)^{-1}$                      \\
    &      $\bomega_{t \C t}=\bomega_{t \C t-1}+\bm{K}_t(\Y_t-\B_t\bomega_{t \C t-1})$              \\
 &  $\V_{t \C t}=\V_{t \C t-1}-\bm{K}_t\B_t \V_{t \C t-1}$    \\ \hline
\multicolumn{2}{|c|}{Initialize $\V_{T,T-1 \C T}=(I-\bm{K}_T\B_T)\V_{T-1 \C T-1}$}                                                    \\
\multicolumn{2}{|c|}{Repeat the {smoothing step for $t=T,\dots,1$} }                              \\ \hline
\cellcolor[HTML]{C0C0C0}Smoothing Step                    &  $\bomega_{t-1 \C T}= \bomega_{t-1 \C t-1}+\bm{Z}_{t-1}(\bomega_{t \C T}-\bomega_{t \C t-1})$                 \\
&  $\V_{t-1 \C T}=\V_{t-1 \C t-1}+\bm{Z}_{t-1}(\V_{t \C T}-\V_{t \C t-1})\bm{Z}_{t-1}'$                      \\
&  $\V_{t,t-1 \C T} = \V_{t-1 \C t-1}\bm{Z}_{t-2}'+\bm{Z}_{t-1}(\V_{t,t-1 \C T}-\V_{t-1 \C t-1})\bm{Z}_{t-2}'$                     \\
& where     $\bm{Z}_{t-1}=\V_{t-1 \C t-1}\V_{t \C t-1}^{-1}$              \\ \hline
\end{tabular}}
\end{center}
\caption{\small Kalman Filter and Smoother Algorithm for Parameter Expanded EM using rotated loading matrices $\bm{B}_{1:T}$}
\label{kalman}
\end{table}}

\vspace{-30pt}
\begin{align*}
- Q_3(\A_{1:T}) &= \frac{1}{2} \sum_{t=1}^T \log |\A_t| +\frac{1}{2} tr \{\A_t^{-1}(\bm{M}_{1t}-\bm{M}_{12t}-\bm{M}_{12t}'+\bm{M}_{2t})\},
\end{align*}
where 
\begin{align*}
\bm{M}_{1t} & = (\bomega_{t-1 \C T}\bomega_{t-1 \C T}'+\V_{t-1 \C T}), \\
\bm{M}_{12t} & =  (\bomega_{t-1 \C T}\bomega_{t \C T}'+\V_{t,t-1 \C T}), \\
\bm{M}_{2t} & =  (\bomega_{t \C T}\bomega_{t \C T}'+\V_{t \C T}).
\end{align*}

}

\subsection{Derivation of the M-step}\label{sec:mstep}

In the M-step, we optimize the function $Q_1(\cdot)$   with respect to   $\B_{0:T}^\star$, given values of $\bSigma_{1:T}$ from the previous M-step. Given  the new values $\B_{0:T}^{\star(m+1)}$ and the posterior moment estimates of the latent factors obtained from the Kalman filter, we optimize $Q_1(\cdot)+Q_2(\cdot)$, with respect to $\bSigma_{1:T}$. {Finally, we optimize the function $Q_3(\cdot)$  with respect to $\A_{1:T}$}.

Optimizing $Q_1(\cdot)$ with respect to  $\B^\star_{0:T}$  boils down to solving a series of independent dynamic spike and slab LASSO regressions (similarly as in\citep{rockova2018dynamic}).  This is justified by the following lemma.

\begin{lemma}\label{A1}
Let $\Y^t=(Y^t_1,\dots, Y^t_P)' \in \mathbb{R}^{P}$ denote the snapshot of the series  at time $t$ and
for $1\leq j\leq P$ define a zero-augmented response vector at time $t$ with $\wt{\Y}^t_j=(Y^t_j,\underbrace{0,\dots,0}_{K})'\in\R^{K+1}$.
For the SVD decomposition $\V_{t \C T}= \sum_{k=1}^K s_k \bm{U}_k^t (\bm{U}_k^t)'$, we denote  with $\wt{\bm{U}}_k^t=\sqrt{s_k} \bm{U}_k^t$ and 
with $\bm\Omega^t=[\bomega_{t|T},\wt{\bm{U}}_1^t, \dots,\wt{\bm{U}}_K^t]' \in \mathbb{R}^{(1+K) \times K}$ and we let $\bm{\beta}^{t \star' }_j \in \mathbb{R}^K$ be the $j^{th}$ row of $\bm{B}_t^{\star}$.
Then we can decompose
$$
Q_1(\bm{B}^\star_{0:T}\C\bSigma_{1:T})=C+
 \sum_{j=1}^P \left[Q_j(\bm{\beta}^{t \star}_j)+Q^0(\bm{\beta}^{0 \star}_j) +\wt Q(\bm{\beta}^{1 \star}_j,\dots,\bm{\beta}^{T \star}_j)\right],
$$
where
\begin{align*}
Q^0(\bm{\beta}^{0 \star}_j) & =\sum_{k=1}^K \left[ \frac{\langle \gamma_{jk}^0 \rangle (\beta_{jk}^{0*})^2}{2\lambda_1/(1-\phi^2)}+ (1-\langle \gamma_{jk}^0 \rangle) | \beta_{jk}^{0*} | \lambda_{0} \right] \\
Q_j(\bm{\beta}^{t \star}_j) & = 
\sum_{t=1}^T\left[\frac{1}{2}\log \sigma_{jt}^2
+ \frac{1}{2\sigma_{jt}^2} ||\wt{\Y}^t_j- {\bm\Omega}^t \bm{\beta}^{t \star}_j||^2_2\right]\\
\wt Q(\bm{\beta}^{1 \star}_j,\dots,\bm{\beta}^{T \star}_j)&=\sum_{t=1}^T \sum_{k=1}^K \left[\frac{ \langle \gamma_{jk}^t \rangle (\beta_{jk}^{t*}-\phi \beta_{jk}^{t-1*})^2}{2\lambda_1}
+(1-\langle \gamma_{jk}^t \rangle) | \beta_{jk}^{t*} | \lambda_0 \right].
\end{align*} 
\end{lemma}

\begin{proof}
Denote with
\begin{equation*}
L\equiv tr\left\{\frac{1}{2} \sum_{t=1}^T \bm{\Sigma}_t^{-1} \left[(\Y_t-\B_t^\star \bomega_{t|T})(\Y_t-\B_t^\star \bomega_{t|T})'+\B_t^\star \V_{t|T} \B^{\star'}_t\right]\right\}.
\end{equation*}
Because
$\bm{B}_t^{\star} \V_{t \C T} \bm{B}^{\star'} _t = \bm{B}_t^\star \sum_{k=1}^K s_k \bm{U}_k^t \bm{U}_k^{t'} (\bm{B}^{\star}_t)' = \sum_{k=1}^K (\bm{0}-\bm{B}^\star_t \wt{\bm{U}}_k^t)(\bm{0}-\bm{B}^\star_t \wt{\bm{U}}_k^t)',$
we have
\begin{equation*}
tr\left\{ \bm{\Sigma}_t^{-1} \bm{B}_t^{\star} \V_{t \C T} \bm{B}^{\star'} _t \right \} = \sum_{k=1}^K (\bm{0}-\bm{B}^\star_t \wt{\bm{U}}_k^t)'\bm{\Sigma}_t^{-1}(\bm{0}-\bm{B}^\star_t \wt{\bm{U}}_k^t).
\end{equation*}
Since $\bSigma_t= \mathrm{diag}(\sigma_{1t}^2, \dots, \sigma_{Pt}^2)$, we have 
\begin{align*}
L& =\frac{1}{2} \sum_{j=1}^P \sum_{t=1}^T  \left[\frac{(Y^t_j- \bomega_{t\C T}' \bm{\beta}^{t \star}_j)^2} {\sigma_{jt}^2} + \sum_{k=1}^K \frac{({0} - \wt{\bm{U}}_k^{t\prime} \bm{\beta}^{t\star}_j)^2} {\sigma_{jt}^2}\right] \\
& =  \sum_{j=1}^P \sum_{t=1}^T \frac{1}{2\sigma_{jt}^2} ||\wt{\Y}^t_j- {\bm\Omega}^t \bm{\beta}^{t \star}_j||^2_2.\qedhere
\end{align*}
\end{proof}

Each summand $Q_j(\bm{\beta}^{t \star}_j)+Q^0(\bm{\beta}^{0 \star}_j) +\wt Q(\bm{\beta}^{1 \star}_j,\dots,\bm{\beta}^{T \star}_j)$ corresponds to a penalized dynamic regression with $K+1$ observations at each time $t$.
Given $\bSigma_t$, finding $\B^{\star(m+1)}$ thereby reduces to solving these $J$ individual regressions. As shown in \cite{rockova2018dynamic}, each regression can be decomposed into a sequence of univariate optimization problems. We use the one-step late EM variant in \cite{rockova2018dynamic}  to obtain closed form one-site updates for each  $\beta_{jk}^{\star t}$ for $(j,k,t)$. Note that this corresponds to a generalized EM, which is aimed at improving the objective relative to the last iteration (not necessarily maximizing it).

These univariate updates are  slightly different from  \cite{rockova2018dynamic}, because we now have $K+1$ observations at time $t$, not just one.
Denote with $\widehat{\beta}_{jl}^{*t}$ the most recent update of the coefficient $\beta_{jl}^{\star t}$.
Let 
 $$
 z_{jk}^{t}=\frac{1}{\sigma_{jt}^2}\sum_{r=1}^{K+1} (\wt Y_{jr}^{t}- \sum_{l \neq k} \wt{\omega}_{rl}^{t} \widehat{\beta}_{jl}^{t*}) \wt{\omega}_{rk}^{t}
 $$
 and denote 
 $$
 Z_{jk}^{t}=z_{jk}^{t} + \frac{\langle \gamma_{jk}^{t} \rangle \phi_1}{\lambda_1} \widehat{\beta}_{jk}^{t-1}+ \frac{\langle \gamma_{jk}^{t+1}\rangle \phi_1}{\lambda_1} \widehat{\beta}_{jk}^{t+1}
 $$ 
 and 
 $$
 W_{jk}^{t}=\frac{1}{\sigma_{jt}^2}\sum_{r=1}^{K+1}(\wt{\omega}^t_{rk})^{2} + \frac{\langle \gamma_{jk}^t \rangle}{\lambda_1} + \frac{\langle \gamma_{jk}^{t+1} \rangle \phi_1^2}{\lambda_1}. $$
 
Then from the calculations in Section 6 of  \cite{rockova2018dynamic} (equations (30)-(33)) we obtain the following update for $\widehat{\beta}_{jk}^{*t}$:
\begin{equation}\label{update_betas}
{\beta}_{jk}^{t\star(m+1)} = 
\begin{cases}
\frac{1}{W_{jk}^t+(1-\phi_1^2)/\lambda_1 M_{jk}^t}[Z_{jk}^{t}- \Lambda_{jk}^t]_{+} \mathrm{sign}(Z_{jk}^{t}) & \quad \text{for}\quad1 < t < T\\
  \frac{1}{\langle \gamma_{jk}^1 \rangle \phi_1^2 + \langle \gamma_{jk}^0 \rangle (1-\phi_1^2)}[\langle \gamma_{jk}^0 \rangle \wh{\beta}_{jk}^1 \phi_1 - (1-\langle \gamma_{jk}^0 \rangle)\lambda_0 \lambda_1]_{+}\mathrm{sign}(\wh{\beta}_{jk}^1) & \quad \text{for}\quad t=0
\end{cases}
\end{equation}
where $M_{jk}^t =\langle \gamma_{jk}^{t+1} \rangle (1- \theta_{jk}^{t+1}) - (1-\langle \gamma_{jk}^{t+1} \rangle) \theta_{jk}^{t+1}$ and $\Lambda_{jk}^t = \lambda_0[(1- \langle \gamma_{jk}^t \rangle) - M_{jk}^t]$.

Given $\B^{\star(m+1)}$, optimizing $Q_1(\cdot)+Q_2(\cdot)$ with respect to $\bSigma_{1:T}$ is done using the Forward Filtering Backward Smoothing algorithm \citep[Ch.~4.3.7][]{Prado2010}. 
In order to maximize the posterior log likelihood with respect to $\bSigma_{1:T}$, we first estimate the parameters of the posterior distribution $\pi(\bSigma_{1:T}\C \Omega,\Y)$, given the updated factor loading matrices $B_{1:T}$,  and then calculate the mode of the posterior.
Although the exact analytical posterior is unattainable, a fast Gamma approximation exists \citep[Ch.~10.8][]{WestHarrison1997book2}. Appropriate Gamma approximations to the posterior have the form 
$$
\pi(1/\sigma_{j,T-k}^2 \C \Omega, \Y) = \bm{G}[\eta_{jT}(-k)/2,d_{jT}(-k)/2],
$$
where $d_{jT}(-k)= \eta_{jT}(-k)s_{jT}(-k)$, with
$$s_{jT}(-k)^{-1}=(1-\delta)s_{j,T-k}^{-1}+\delta s_{jT}(-k+1)^{-1}$$,
and filtered degrees of freedom defined by
$$
\eta_{jT}(-k)= (1-\delta)\eta_{j,T-k}+\delta \eta_{j,T-k+1},
$$
initialized at $\eta_{jT}(0)=\eta_{jT}$. Here $s_{j,T-k}$ denotes $\mathbb{E}(\sigma_{j,T-k}^2 \mid \Omega_{T-k}, \Y_{T-k})$. The details of the algorithm is given in Algorithm \ref{ffbs}. In the algorithm we denote the diagonal matrices with diagonal entries $\eta_{j,T-k}$ by $\bm{\eta}_{T-k}$ and analogously define matrices $\bm{D}_{T}(-k)$, $\bm{S}_{T-k}$ and $S_{T}(-k)$ for $k=0,1,\dots,T-1$ so that we can update the parameters of the posterior distribution simultaneously for all $j$ and fixed $t$.
In our study, we set the prior degrees of freedom $\eta_0$ to its limit  $\eta_0=(1-\delta)^{-1}$ in order to achieve stability and efficiency.
Given the parameters of the posterior distribution (the expectation and degrees of freedom), computing the posterior mode is straight forward.

{
Finally, the updates for the covariance matrices $\A_{1:T}$, obtained by maximizing $Q_3(\cdot)$, have the following closed form
\begin{align*}
{\A_t^{(m+1)}} & = \bm{M}_{1t}-\bm{M}_{12t}-\bm{M}_{12t}'+\bm{M}_{2t} \quad for \quad t=1,\dots,T.
\end{align*}
}
After completing the expanded M-step in the $(m+1)^{st}$ iteration, we perform a rotation step 
towards the reduced parameter space to obtain
$$
{\bm{B}_t}^{(m+1)}={\bm{B}_t^{\star(m+1)}}{\A_{tL}}^{(m+1)},
$$
where ${\A_t}^{(m+1)}={\A_{tL}^{(m+1)}}{\A_{tL}}^{(m+1)'}$ is the Cholesky decomposition.  These rotated factor loading matrices are carried forward to the next E-step, where we again use the reduced parameter form by keeping $\A_t=\sigma^2_\omega \bm I$.

{
\spacingset{1}
\begin{table}[t]
\small
\begin{center}
\scalebox{0.8}{
\begin{tabular}{|l|l|}
\hline
\multicolumn{2}{|c|}{\cellcolor[HTML]{C0C0C0} \textbf{Algorithm:} \textit{Forward Filtering Backward Smoothing}} \\ \hline \hline
\multicolumn{2}{|c|}{Input: $\B_{1:T}$ and $\bSigma_{1:T}$ from previous iteration}                                                    \\
\multicolumn{2}{|c|}{Initialize $\bm{\eta}_0$, $\bm{D}_0$, $\bm{S}_0=\bm{D}_0 \bm{\eta}_0^{-1}$}                                                    \\
\multicolumn{2}{|c|}{Repeat the Forward Step for $t=1,\dots,T$}                                \\ \hline
\cellcolor[HTML]{C0C0C0}Forward Step                      & $\bm{\eta}_t=\delta \bm{\eta}_{t-1}+\mathrm{I}$                     \\
& $\bm{D}_t=\delta \bm{D}_{t-1} + \bm{S}_{t-1}\bm{E}_t\bm{E}_t'\bm{Q}_t^{-1}$                       \\
& $\bm{S}_t=\bm{D}_t \bm{\eta}_t^{-1}$ \\
 where & $\bm{E}_t=\Y_t-\B_t \bomega_{t \C t-1}$\\
 & $\bm{Q}_t= \B_t'\V_{t \C t-1}\B_t + \bSigma_t $\\ \hline
\multicolumn{2}{|c|}{Initialize $\bm{S}_T(0)=S_T$}                                                    \\
\multicolumn{2}{|c|}{{Repeat the Backward Step for $k=1,\dots,T-1$}}                               \\ \hline
\cellcolor[HTML]{C0C0C0}Backward Step      & $\bm{\eta}_T(-k)=(1-\delta) \bm{\eta}_{T-k} + \delta \bm{\eta}_{T-k+1}$ \\
& $\bm{S}_T(-k)^{-1}= (1-\delta) \bm{S}_{T-k}^{-1} + \delta \bm{S}_T (-k+1)^{-1}$                       \\
& $\bm{D}_T(-k) = \bm{\eta}_T(-k) \bm{S}_T(-k)$                       \\ 
& $\bUpsilon_{T-k} = (\bm{\eta}_T(-k) - \mathrm{I})\bm{D}_T(-k)^{-1}$ \\ \hline
\cellcolor[HTML]{C0C0C0}Compute Mode                     & $\bSigma_{T-k}=\bUpsilon_{T-k}^{-1}$                       \\ \hline
\end{tabular}}
\end{center}
\caption{\small Forward Filtering Backward Smoothing algorithm for estimating idiosyncratic variances.}
\label{ffbs}
\end{table}}

\section{Appendix: B}
\renewcommand{\thefigure}{B\arabic{figure}}\setcounter{figure}{0}
\renewcommand{\thetable}{B\arabic{table}}\setcounter{table}{0}

\subsection{Additional Tables and Graphs}\label{app:tables}

\begin{sidewaystable}[htbp!]
\setlength\tabcolsep{2.5pt}
\scriptsize
\begin{tabular}{llllllllllllllll}
\hline
   & \begin{tabular}[c]{@{}l@{}}Output \\ and \\ Income\end{tabular} &    & Labor Market  &    & \begin{tabular}[c]{@{}l@{}}Consumption \\ and \\ Orders\end{tabular} &    & \begin{tabular}[c]{@{}l@{}}Orders \\ and \\ Inventories\end{tabular} &    & \begin{tabular}[c]{@{}l@{}}Money \\ and \\ Credit\end{tabular} &     & \begin{tabular}[c]{@{}l@{}}Interest rate \\ and \\ Exchange Rates\end{tabular} &     & Prices          &     & Stock Market   \\ \hline
1  & RPI                                                             & 17 & HWI           & 48 & HOUST                                                                & 58 & DPCERA3M086SBEA                                                      & 67 & M1SL                                                           & 81  & FEDFUNDS                                                                       & 103 & WPSFD49207      & 123 & S\&P 500       \\
2  & W875RX1                                                         & 18 & HWIURATIO     & 49 & HOUSTNE                                                              & 59 & CMRMTSPLx                                                            & 68 & M2SL                                                           & 82  & CP3Mx                                                                          & 104 & WPSFD49502      & 124 & S\&P: indust   \\
3  & INDPRO                                                          & 19 & CLF16OV       & 50 & HOUSTMW                                                              & 60 & RETAILx                                                              & 69 & M2REAL                                                         & 83  & TB3MS                                                                          & 105 & WPSID61         & 125 & S\&P div yield \\
4  & IPFPNSS                                                         & 20 & CE16OV        & 51 & HOUSTS                                                               & 61 & AMDMNOx                                                              & 70 & AMBSL                                                          & 84  & TB6MS                                                                          & 106 & WPSID62         & 126 & S\&P PE ratio  \\
5  & IPFINAL                                                         & 21 & UNRATE        & 52 & HOUSTW                                                               & 62 & ANDENOx                                                              & 71 & TOTRESNS                                                       & 85  & GS1                                                                            & 107 & OILPRICEx       & 127 & VXOCLSx        \\
6  & IPCONGD                                                         & 22 & UEMPMEAN      & 53 & PERMIT                                                               & 63 & AMDMUOx                                                              & 72 & NONBORRES                                                      & 86  & GS5                                                                            & 108 & PPICMM          &     &                \\
7  & IPDCONGD                                                        & 23 & UEMPLT5       & 54 & PERMITNE                                                             & 64 & BUSINVx                                                              & 73 & BUSLOANS                                                       & 87  & GS10                                                                           & 109 & CPIAUCSL        &     &                \\
8  & IPNCONGD                                                        & 24 & UEMP5TO14     & 55 & PERMITMW                                                             & 65 & ISRATIOx                                                             & 74 & REALLN                                                         & 88  & AAA                                                                            & 110 & CPIAPPSL        &     &                \\
9  & IPBUSEQ                                                         & 25 & UEMP15OV      & 56 & PERMITS                                                              & 66 & UMCSENTx                                                             & 75 & NONREVSL                                                       & 89  & BAA                                                                            & 111 & CPITRNSL        &     &                \\
10 & IPMAT                                                           & 26 & UEMP15T26     & 57 & PERMITW                                                              &    &                                                                      & 76 & CONSPI                                                         & 90  & COMPAPFFx                                                                      & 112 & CPIMEDSL        &     &                \\
11 & IPDMAT                                                          & 27 & UEMP27OV      &    &                                                                      &    &                                                                      & 77 & MZMSL                                                          & 91  & TB3SMFFM                                                                       & 113 & CUSR0000SAC     &     &                \\
12 & IPNMAT                                                          & 28 & CLAIMSx       &    &                                                                      &    &                                                                      & 78 & DTCOLNVHFNM                                                    & 92  & TB6SMFFM                                                                       & 114 & CUSR0000SAD     &     &                \\
13 & IPMANSICS                                                       & 29 & PAYEMS        &    &                                                                      &    &                                                                      & 79 & DTCTHFNM                                                       & 93  & T1YFFM                                                                         & 115 & CUSR0000SAS     &     &                \\
14 & IPB51222S                                                       & 30 & USGOOD        &    &                                                                      &    &                                                                      & 80 & INVEST                                                         & 94  & T5YFFM                                                                         & 116 & CPIULFSL        &     &                \\
15 & IPFUELS                                                         & 31 & CES1021000001 &    &                                                                      &    &                                                                      &    &                                                                & 95  & T10YFFM                                                                        & 117 & CUSR0000SA0L2   &     &                \\
16 & CUMFNS                                                          & 32 & USCONS        &    &                                                                      &    &                                                                      &    &                                                                & 96  & AAAFFM                                                                         & 118 & CUSR0000SA0L5   &     &                \\
   &                                                                 & 33 & MANEMP        &    &                                                                      &    &                                                                      &    &                                                                & 97  & BAAFFM                                                                         & 119 & PCEPI           &     &                \\
   &                                                                 & 34 & DMANEMP       &    &                                                                      &    &                                                                      &    &                                                                & 98  & TWEXMMTH                                                                       & 120 & DDURRG3M086SBEA &     &                \\
   &                                                                 & 35 & NDMANEMP      &    &                                                                      &    &                                                                      &    &                                                                & 99  & EXSZUSx                                                                        & 121 & DNDGRG3M086SBEA &     &                \\
   &                                                                 & 36 & SRVPRD        &    &                                                                      &    &                                                                      &    &                                                                & 100 & EXJPUSx                                                                        & 122 & DSERRG3M086SBEA &     &                \\
   &                                                                 & 37 & USTPU         &    &                                                                      &    &                                                                      &    &                                                                & 101 & EXUSUKx                                                                        &     &                 &     &                \\
   &                                                                 & 38 & USWTRADE      &    &                                                                      &    &                                                                      &    &                                                                & 102 & EXCAUSx                                                                        &     &                 &     &                \\
   &                                                                 & 39 & USTRADE       &    &                                                                      &    &                                                                      &    &                                                                &     &                                                                                &     &                 &     &                \\
   &                                                                 & 40 & USFIRE        &    &                                                                      &    &                                                                      &    &                                                                &     &                                                                                &     &                 &     &                \\
   &                                                                 & 41 & USGOVT        &    &                                                                      &    &                                                                      &    &                                                                &     &                                                                                &     &                 &     &                \\
   &                                                                 & 42 & CES0600000007 &    &                                                                      &    &                                                                      &    &                                                                &     &                                                                                &     &                 &     &                \\
   &                                                                 & 43 & AWOTMAN       &    &                                                                      &    &                                                                      &    &                                                                &     &                                                                                &     &                 &     &                \\
   &                                                                 & 44 & AWHMAN        &    &                                                                      &    &                                                                      &    &                                                                &     &                                                                                &     &                 &     &                \\
   &                                                                 & 45 & CES0600000008 &    &                                                                      &    &                                                                      &    &                                                                &     &                                                                                &     &                 &     &                \\
   &                                                                 & 46 & CES2000000008 &    &                                                                      &    &                                                                      &    &                                                                &     &                                                                                &     &                 &     &                \\
   &                                                                 & 47 & CES3000000008 &    &                                                                      &    &                                                                      &    &                                                                &     &                                                                                &     &                 &     &                \\ \hline
\end{tabular}
\caption{Macroeconomic Study:  The list of economic variables used in the study.
\label{tab:macrovar}}
\end{sidewaystable}

%

\end{document}